\documentclass[prd,twocolumn,preprintnumbers,nofootinbib,superscriptaddress,amsmath,amssymb,
aps]{revtex4} 

\usepackage{mathrsfs}
\usepackage{graphicx}
\usepackage{dcolumn}
\usepackage{bm}

\usepackage{comment}
\usepackage{braket}
\usepackage{slashed}
\usepackage{amsthm}
\newtheorem*{theorem}{Theorem}

\begin{document}

\preprint{CHIBA-EP-246, 2021.05.21}

\title{Effects of a quark chemical potential on the analytic structure of the gluon propagator}

\author{Yui Hayashi}
\email{yhayashi@chiba-u.jp}
\affiliation{
Department of Physics, Graduate School of Science and Engineering, Chiba University, Chiba 263-8522, Japan
}

\author{Kei-Ichi Kondo}
\email{kondok@faculty.chiba-u.jp}
\affiliation{
Department of Physics, Graduate School of Science and Engineering, Chiba University, Chiba 263-8522, Japan
}
\affiliation{
Department of Physics, Graduate School of Science, Chiba University, Chiba 263-8522, Japan
}


\begin{abstract}
We perform complex analyses of the gluon propagator at nonzero quark chemical potential in the long-wavelength limit, using an effective model with a gluon mass term of the Landau-gauge Yang-Mills theory, which is a Landau-gauge limit of the Curci-Ferrari model with quantum corrections being included within the one-loop level. We mainly investigate complex poles of the gluon propagator, which could be relevant to confinement.
Around typical values of the model parameters, we show that the gluon propagator has one or two pairs of complex conjugate poles depending on the value of the chemical potential. In addition to a pair similar to that in the case of zero chemical potential, a new pair appears near the real axis when the chemical potential is roughly between the effective quark mass and the effective gluon mass of the model. We discuss possible interpretations of these poles. Additionally, we prove the uniqueness of analytic continuation of the Matsubara propagator to a class of functions that vanish at infinity and are holomorphic except for a finite number of complex poles and singularities on the real axis. 
\end{abstract}

\maketitle

\section{INTRODUCTION}

For a long time, it has been expected that quark degrees of freedom would dominate in a highly dense matter of quantum chromodynamics (QCD) rather than hadrons, although details of the phase structure are still unclear mainly due to the sign-problem \cite{denseQCD}.
Studying the analytic structure of the gluon propagator is of importance to this end since this structure provides information on the in-medium behavior, e.g., whether or not a quasiparticle description is appropriate.
We thus explore the analytic structure of the gluon propagator in this article.

The main difficulty in a continuum approach for the quark matter is the breakdown of the perturbation theory in the infrared QCD.
Indeed, perturbative calculations of dense QCD matter suggest that the quark matter is already strongly correlated at the quark chemical potential $\mu_q \lesssim 1$ GeV \cite{KRV2010}.
Therefore, a method valid in infrared is required to study relatively low density region of the quark matter.

About ten years ago, an effective model of the Landau-gauge Yang-Mills theory has been proposed \cite{TW10,TW11} to understand recent numerical lattice results that support the decoupling (massivelike) solution of the Dyson-Schwinger equation (DSE) \cite{decoupling-analytical, decoupling-lattice}.
This model consists of the Faddeev-Popov Lagrangian and the simple gluon mass term, i.e., the Landau gauge limit of the Curci-Ferrari model \cite{CF76b}, which we call \textit{the massive Yang-Mills model}. This mass deformation could be a consequence of generating the dimension-two gluon condensate \cite{Schaden-Kondo-Shinohara-Wershinke,Boucaud2000,Gubarev-Zakharov,Verschelde2001,BG2003} or avoiding the Gribov ambiguity \cite{Gribov78,ST12}.
The massive Yang-Mills model has the modified Becchi-Rouet-Stora-Tyutin (BRST) symmetry and is multiplicatively renormalizable to be proved through the modified Slavnov-Taylor identities (at all orders of the perturbation theory) \cite{CF76b}.
Moreover, there exists the ``infrared safe'' renormalization scheme respecting the nonrenormalization theorems \cite{Taylor71,Non-mass-ren}, in which
the running gauge coupling constant $g$ is finite at all scales on some renormalization group (RG) flows \cite{TW11,RSTW17,ST12}.

This model provides the gluon and ghost propagators that agree strikingly with the numerical lattice results just in the one-loop level \cite{TW10,TW11}.
The three-point functions \cite{PTW13} and two-point correlation functions at finite temperature \cite{RSTW14} were compared to the numerical lattice results with good accordance. Furthermore, the two-loop corrections improve the agreement for the gluon and ghost propagators \cite{GPRT19}.
Therefore, the effective mass captures some nonperturbative aspects of the Yang-Mills theory.

For unquenched cases with the number of quark flavors $N_F = 2,~2+1+1$, the massive Yang-Mills model with dynamical quarks gives the gluon and ghost propagators consistent with the numerical lattice results as well \cite{PTW14}. For the quark sector, higher loop corrections are important in this model \cite{PTW14, PTW15, PRSTW17}. Also, QCD phases with heavy quarks have been studied in a similar model in the Landau-DeWitt gauge ~\cite{Finite-temperature-Landau-dewitt}.
Despite the shortcoming of the massive Yang-Mills model for describing the quark sector, this model will be useful for the analyses of the gluon propagator.

One might worry about the absence of the nilpotent BRST symmetry. The massive Yang-Mills model certainly suffers from the physical nonunitarity \cite{CF76b, Kondo13} as a consistent theory. However, as this model gives the well-approximating propagator given by a mass-deformation that has some background refereed above, we can still consider the massive Yang-Mills model as a model for describing the gluon and ghost propagators.
Therefore, it is interesting to investigate the analytic structure of the gluon propagator of this model.

In \cite{Suenaga-Kojo19}, the gluon propagator of the massive Yang-Mills model at finite chemical potential (and zero temperature) obtained in the vanishing momentum renormalization scheme, which is not infrared safe, has been compared with numerical lattice data \cite{BHMS19} for the gauge group $SU(2)$, and the propagators at finite temperature and density have been compared in \cite{KS2021}.
While the singlet diquark gap can improve the consistency with the lattice results, the agreement with the lattice results is not quite satisfactory for parameters $(g,M)$ that are independent of chemical potential.
If one enables the gluon mass parameter $M$ to depend on the chemical potential, one can obtain a fair agreement between the massive Yang-Mills model and the numerical lattice results.
Therefore, although this model may lack some important aspects, it is still worthwhile studying the analytic structure of the gluon propagator at a finite chemical potential by utilizing the massive Yang-Mills model with various model parameters.

In the vacuum case, i.e., vanishing chemical potential $\mu_q = 0$, we have investigated the analytic structures of the gluon, quark, and ghost propagators and revealed that the gluon and quark propagators have one pair of complex conjugate poles while the ghost propagator has no complex poles \cite{HK2018,KWHMS19,HK2020}.  Other several models and reconstruction methods also predict such complex poles of the gluon propagator, e.g. \cite{Gribov78,Zwanziger90,DGSVV2008,Siringo16a,Siringo16b,Stingl86,Stingl96,HKRSW, BT2019, Falcao:2020vyr}. The DSE with the ray technique had provided the gluon propagator holomorphic except for timelike momenta \cite{SFK12}, but the recent study \cite{Fischer-Huber} has updated this conclusion and strongly suggested a singularity on the complex momentum plane.

Complex poles invalidate the K\"all\'en-Lehmann spectral representation \cite{spectral_repr_UKKL} and might correspond to unphysical degrees of freedom in an indefinite metric state space \cite{Nakanishi72suppl}. Therefore, the complex poles represent deviations from observable particles and are expected to be related to the confinement mechanism. For example, the connection between complex poles of the fermion propagator and confining potential in three-dimensional quantum electrodynamics has been discussed in \cite{Maris}. Incidentally, another generalization of the spectral representation taking unphysical degrees of freedom into account is proposed in \cite{Lowdon}.

In this article, we investigate the analytic structure of the in-medium gluon propagator at finite quark chemical potential $\mu_q$ by employing the massive Yang-Mills model with quantum corrections being included within the one-loop level. Since we are interested in the long-distance behavior and analytic structure of the gluon propagator on the complex frequency plane, we perform complex analyses on the gluon propagator in \textit{the long-wavelength limit} $\vec{k} \rightarrow 0$. In addition, we consider the uniqueness of the analytic continuation in the presence of complex poles, since we use the Matsubara propagator in the low-temperature limit $T \rightarrow + 0$ and the analytic continuation is in principle not unique before taking the limit.

This article is organized as follows. In the next section, the definition of complex poles of an in-medium propagator and the method for counting complex poles are presented. A proof of the uniqueness in a class of functions having a finite number of complex poles is provided in Appendix A. The massive Yang-Mills model and its one-loop expressions are presented in Sec. III. We detail the vacuum part of the one-loop expressions in Appendix B. In Sec. IV and Appendix C, we determine the number of complex poles and their locations in the space of the model parameters and the spectral function at a specific set of model parameters.  It turns out that the gluon propagator has one pair of almost real poles in addition to the other pair of complex conjugate poles similar to the vacuum ones at intermediate quark chemical potential. In Sec.~V, we discuss possible interpretations of these almost real poles, an improvement of the results, estimates for slightly large $\mu_q$, and infrared problems at finite temperature.
In Sec.~VI, a summary of these findings and future prospects are given.

\section{Complex poles of in-medium propagators}

In this section, we define complex poles of propagators in medium. Then, we introduce a method to count the number of complex poles, which is utilized in the subsequent sections.

\subsection{Definition}

 \begin{figure}[tb]
  \begin{center}
   \includegraphics[width=\linewidth]{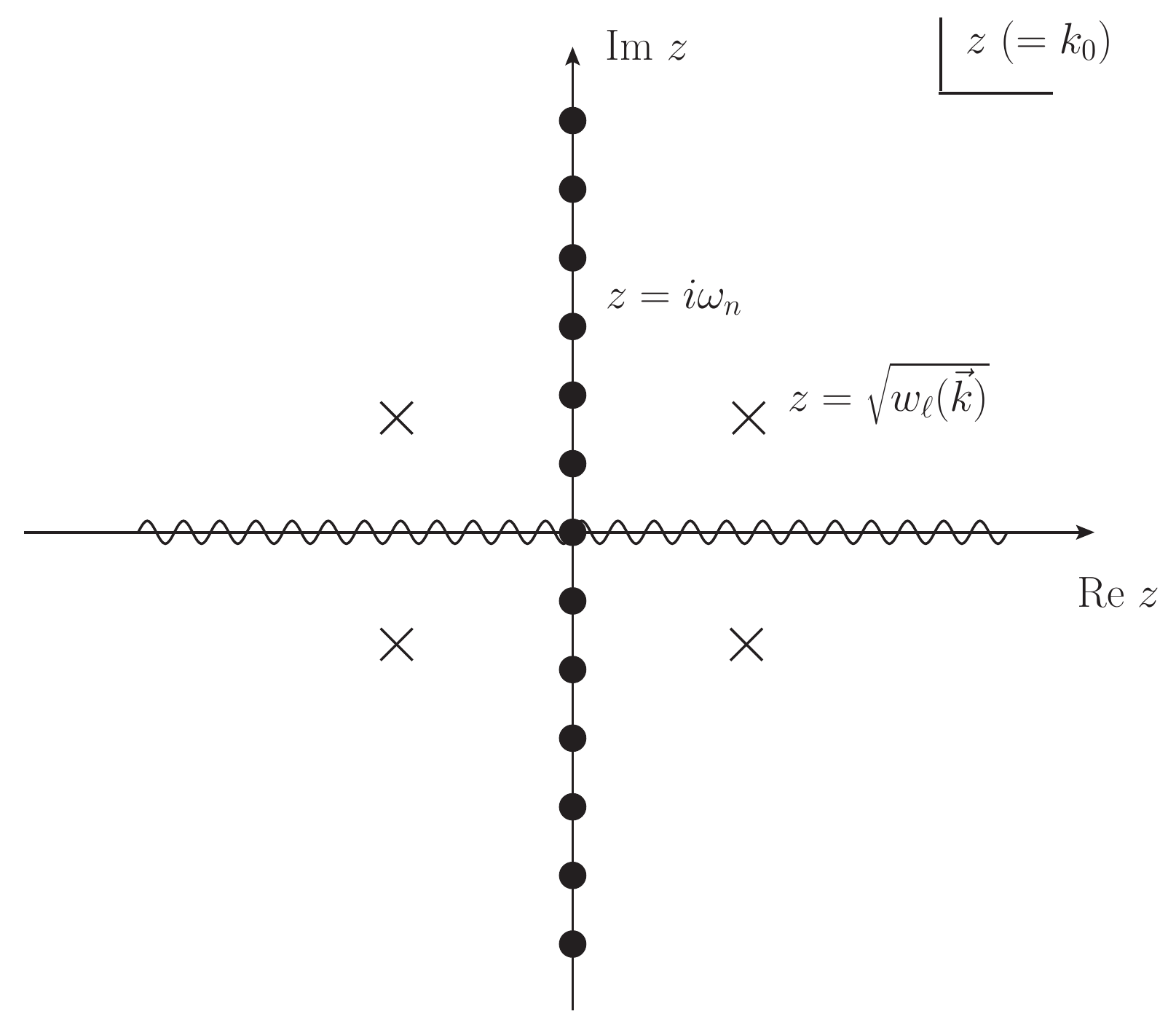}
  \end{center}
   \caption{Schematic picture of singularities on the complex $z (=k_0)$ plane. We analytically continue a Matsubara propagator $D(z = i \omega_n, \vec{k})$ defined at the Matsubara frequencies $z = i \omega_n$ (shown as the dots) to $D(z, \vec{k})$ on the complex $z$ plane. The two sets of complex conjugate poles in the $z$ plane represent a pair of complex conjugate poles with respect to $z^2$ in (\ref{eq:generalized-spec-repr}).}
    \label{fig:section2_Matsubara}
\end{figure}
In medium, we compute a Matsubara propagator $D(i \omega_n, \vec{k})$ within the imaginary-time formalism, where $\omega_n$ is the Matsubara frequency and $\vec{k}$ is the spatial momentum. We consider the analytic continuation $D(z, \vec{k})$ on the complex $z$ plane for a fixed $\vec{k}$ from the Matsubara frequencies on the imaginary axis $z = i \omega_n$.
This provides information on the spectrum and is useful for studying linear response, in which the retarded propagator, namely the propagator analytically continued to the real axis from the upper-half plane, plays an important role \cite{Kapusta-Gale}.

For a field describing a physical observable particle, the usual spectral representation holds. The spectral condition forces analytically-continued Matsubara propagator $D(z, \vec{k})$ to have singularities only on real axis $z \in \mathbb{R}$.

However, the spectral condition may be violated for confined degrees of freedom, since not all states have to be physical. Thus, we can consider the possibility of complex spectra, which need not be excluded in an indefinite metric state space \cite{Nakanishi72suppl}. If a state with complex energy exists, this should correspond to a confined state. Further formal aspects will be discussed elsewhere \cite{HK2021}.

Here, we assume the following generalized spectral representation allowing complex poles for the gluon propagator $D(z,\vec{k})$, which is a propagator obtained by the analytic continuation from the Matsubara propagator $D(i \omega_n,\vec{k})$ defined at points on the pure imaginary axis of the complex $z$ plane:
\begin{align}
    D(z,\vec{k}) &= \int_0 ^\infty d \sigma^2 \frac{\rho(\sigma,\vec{k})}{\sigma^2 - z^2} + \sum_{\ell =1}^n \frac{Z_\ell(\vec{k})}{ w_\ell(\vec{k}) - z^2}, \label{eq:generalized-spec-repr} \\
    \rho(\sigma,\vec{k}) &= \frac{1}{\pi} \operatorname{Im} D(\sigma + i \epsilon,\vec{k}),
\end{align}
where $\rho(\sigma,\vec{k})$ is the spectral function, $w_\ell (\vec{k}) $ is a position of a complex pole, and $Z_\ell(\vec{k})$ is its residue for arbitrary but fixed $\vec{k}$.
Figure \ref{fig:section2_Matsubara} illustrates singularities on the complex $z$ plane of $D(z,\vec{k})$.

Notice that, in the vacuum case, there is a one-to-one correspondence between the propagator $D(z,\vec{k})$ analytically continued to the upper-half plane in $z$ and the analytic continuation in the complex $k^2$ plane $\tilde{D}(k^2)$, which has been considered in the previous articles \cite{HK2018, HK2020}, in the sense that $D(z,\vec{k}) = \tilde{D} (z^2 - \vec{k}^2)$.

Since the set of Matsubara frequencies $\{ \omega_n \}$ has no accumulation points, uniqueness of the analytic continuation is an important problem to be proved.
Indeed, there is a well-known theorem saying that the uniqueness holds in a class of functions satisfying (i) $D(z) \rightarrow 0$ as $|z| \rightarrow \infty$ and (ii) $D(z)$ is holomorphic except for the real axis, i.e., these two conditions are sufficient to determine the unique continuation \cite{Baym-Mermin}. Although this theorem cannot be applied to our case due to the existence of complex poles, we can generalize this theorem in a straightforward way. In Appendix A, we present a proof of the uniqueness under the weaker conditions allowing complex poles:
\begin{enumerate}
    \item $D(z) \rightarrow 0$ as $|z| \rightarrow \infty$,
    \item $D(z)$ is holomorphic except for singularities on the real axis and a finite number of complex poles. 
\end{enumerate}
Therefore, the uniqueness of the analytic continuation from the Matsubara propagator is valid in a similar sense even in the presence of complex poles.

Note that complex poles defined here do not correspond to poles of quasiparticles. This is because the complex poles defined here yield poles in both of the upper-half and lower-half planes in $z$. While a quasiparticle pole is in the second Riemann sheet in $z^2$, the complex pole is in the first Riemann sheet.

\subsection{Counting complex poles}

Let us introduce a method to count the number of complex poles based on the argument principle \cite{HK2018, HK2020} to be used in the following sections. We can relate a propagator at real frequencies to complex poles and zeros.

In the vacuum case, we have applied the method to a propagator on the complex $k^2$ plane. For an in-medium propagator, we can take $k^2$ as the squared complex frequency $z^2$. The statement is as follows. 

Suppose that a complex-valued propagator $D(z^2) := D(z,\vec{k})$ with a fixed spatial momentum $\vec{k}$ and its data $\{ D(z^2 = x_n + i \epsilon) \}$ for real frequencies $z$ (namely, $z^2 > 0$) satisfy the following conditions.
\begin{enumerate}
 \item In the limit $|z^2| \rightarrow \infty$, $D(z^2)$ has the same phase as the free propagator, i.e., $\arg (-D(z^2)) \rightarrow  \arg \frac{1}{z^2} $ as $|z^2| \rightarrow \infty$.
 \item In the limit $|z^2| \rightarrow 0$,  $D(z^2 = 0) > 0$.
  \item The sequence $\{z^2 = x_n + i \epsilon \}_{n=0}^N$ is sufficiently dense so that $D(z^2= x + i \epsilon)$ changes its phase at most half-winding ($\pm \pi$) between $x_n+i\epsilon$ and $x_{n+1}+i\epsilon$, i.e., for $n = 0, 1, \cdots, N$,
 \begin{align}
\left| \int_{x_n}^{x_{n+1}} dx \frac{d}{dx} \arg D(x + i \epsilon) \right| < \pi,
\end{align}
where we denote sufficiently small $x_0 = \delta^2 >0$ and sufficiently large $x_{N+1} = \Lambda^2$, on which we will take the limits $\delta^2 \rightarrow +0$ and $\Lambda^2 \rightarrow +\infty$.
\end{enumerate}
Then the winding number, which is the difference between the number of complex zeros ($N_Z$) and poles ($N_P$) with respect to $z^2$, reads
\begin{align}
N_W (C) &= N_Z - N_P \notag \\
&= -1 + 2 \sum_{n=0}^N \frac{1}{2\pi} \operatorname{Arg}\left[ \frac{D(x_{n+1}+i\epsilon)}{D(x_n + i \epsilon)}\right]. 
\end{align}
Thus the number of complex poles $N_P$ is given by
\begin{align}
N_P &= N_Z -  N_W (C) \notag \\
&= N_Z + 1 - 2 \sum_{n=0}^N \frac{1}{2\pi} \operatorname{Arg}\left[ \frac{D(x_{n+1}+i\epsilon)}{D(x_n + i \epsilon)}\right]. \label{eq:winding-timelike}
\end{align}
For details of the derivation, see \cite{HK2020}.
When the three conditions (i), (ii), and (iii) hold, we can numerically compute the number of complex poles $(N_P)$ from the number of zeros ($N_Z$) and data at the real frequencies $\{ D(x_n + i \epsilon) \}$

Throughout this article, $N_P$ denotes the number of complex poles on the $z^2$ plane, i.e., the number of poles on the (upper-)half plane on the $z$ plane, and ``complex conjugate poles'' denote those on the $z^2$ plane. The propagator has $2 N_P$ complex poles on the whole $z$ complex plane.

\section{Model}

In this section, we introduce the massive Yang-Mills model, which is regarded as an effective model of the Landau-gauge Yang-Mills theory, or the Landau-gauge limit of Curci-Ferrari model, and review the one-loop expressions.

\subsection{Massive Yang-Mills model}

The Euclidean Lagrangian of the model at $N$ colors with $N_F$ flavors is given by \cite{TW10,TW11,PTW14}
\begin{align}
{\mathscr L}_{mYM} &= {\mathscr L}_{YM} + {\mathscr L}_{GF} + {\mathscr L}_{FP} + {\mathscr L}_{m}+ {\mathscr L}_{q}, \\
{\mathscr L}_{YM} &= \frac{1}{4} {\mathscr F}^A_{\mu \nu} {\mathscr F}^{A}_{\mu \nu}, \notag \\
{\mathscr L}_{GF} &= i \mathscr{N}^A \partial_\mu {\mathscr A}^A_\mu, \notag \\
{\mathscr L}_{FP} &=  \bar{{\mathscr C}}^A \partial_\mu {\mathscr D}_\mu[{\mathscr A}]^{AB} {\mathscr C}^B \notag \notag \\
&=  \bar{{\mathscr C}}^A \partial_\mu (\partial_\mu {\mathscr C}^A + g_b f^{ABC} {\mathscr A}^B_\mu {\mathscr C}^C), \notag \\
{\mathscr L}_{m} &= \frac{1}{2} M^2_b {\mathscr A}_\mu^A {\mathscr A}_\mu^A, \notag  \\
{\mathscr L}_{q} &= \sum_{i=1}^{N_F} \bar{\psi}_i (\gamma_\mu {\mathscr D}_\mu[{\mathscr A}] + (m_b)_{q,i}) \psi_i \notag \\
&= \sum_{i=1}^{N_F} \bar{\psi}_i (\gamma_\mu (\partial_\mu - i g_b {\mathscr A}_\mu^A t^A) + (m_b)_{q,i}) \psi_i,
\end{align}
where we have introduced the bare gluon, ghost, anti-ghost, Nakanishi-Lautrup, and quark fields denoted by ${\mathscr A}^A_\mu, ~ {\mathscr C}^A , ~ \bar{{\mathscr C}}^A, \mathscr{N}^A,~(A = 1,2,\cdots, N^2-1)$, and $ \psi_i~(i = 1,2, \cdots, N_F)$ respectively, the bare gauge coupling constant $g_b$, the bare gluon mass $M_b$, and the bare quark mass $(m_b)_{q,i}$, while $f^{ABC}(A,B,C = 1,2,\cdots, N^2-1)$ stand for the structure constants with the generators $t^A$ of the fundamental representation of the group $G = SU(N)$.

The renormalization factors $(Z_A,Z_C, Z_{\bar{C}} = Z_C, Z_{\psi_i}),~ Z_g,~ Z_{M^2}, ~ Z_{m_{q,i}}$ for the gluon, ghost, anti-ghost, and quark fields $({\mathscr A}_\mu, {\mathscr C}, \bar{{\mathscr C}}, \psi_i)$, the gauge coupling constant $g$, and the gluon and quark mass parameters $M^2, m_{q,i}$ are introduced respectively as follows:
\begin{align}
{\mathscr A}^\mu &= \sqrt{Z_A} {\mathscr A}_R^\mu, ~ {\mathscr C} = \sqrt{Z_C} {\mathscr C}_R, \notag \\
\bar{{\mathscr C}} &= \sqrt{Z_C} \bar{{\mathscr C}}_R, ~\psi_i =  \sqrt{Z_{\psi_i}}\psi_{R,i}, \notag \\
~g_b &= Z_g g, ~M^2_b = Z_{M^2} M^2, ~ (m_b)_{q,i} = Z_{m_{q,i}} m_{q,i}.
\end{align}
In this article, we consider the two flavor case $N_F = 2$ and employ this model with degenerate quark masses, $m_q := m_{q,i}$, and therefore $Z_{\psi} := Z_{\psi_i}$ and $Z_{m_{q}} := Z_{m_{q,i}}$.
Notice that the quark mass parameter $m_q$ of this model is chosen to fit the propagators obtained from other methods, e.g., numerical lattice results. In particular, the quark mass parameter $m_q$ is nonzero even for massless quarks due to the spontaneous breakdown of the chiral symmetry.

The general tensorial structure of the gluon propagator ${\mathscr D}_{\mu \nu} (k_E)$ reads, from the spatial rotational symmetry and the transversality of the Landau gauge,
\begin{align}
{\mathscr D}_{\mu \nu} (k_E) = D_T(k_E^2) P^T_{\mu \nu} + D_L(k_E^2) P^L_{\mu \nu}, \label{eq:gluon_propagator_tensorial_str}
\end{align}
where $k_E = (k_1,k_2,k_3,k_4) = (\vec{k}, k_4)$ is the Euclidean momentum, $P^T_{\mu \nu}$ and $P^L_{\mu \nu}$ are the transverse and longitudinal projectors respectively, i.e.,
\begin{align}
P^T_{ij} &= \delta_{ij} - \frac{k_i k_j}{\vec{k}^2}, \notag \\
P^T_{4i} &= P^T_{i4} = P^T_{44} = 0~~ (i,j = 1,2,3),
\end{align}
and,
\begin{align}
P^L_{\mu \nu} = P_{\mu \nu} - P^T_{\mu \nu}, ~~ P_{\mu \nu} = \delta_{\mu \nu} - \frac{k_{E,\mu} k_{E,\nu}}{k_E^2}.
\end{align}

We define the vacuum part of the gluon and ghost two-point vertex functions $\Gamma_{{\mathscr A}, vac}^{(2)}, \Gamma_{gh, vac}^{(2)}$ as the zero temperature $T=0$ and the zero chemical potential $\mu = 0$ limit,
\begin{align}
\left. {\mathscr D}_{\mu \nu} (k_E) \right|_{T=\mu=0} = [\Gamma_{{\mathscr A}, vac}^{(2)} (k_E^2) ]^{-1} P_{\mu \nu}, \notag \\
\left. \Delta_{gh} (k_E) \right|_{T=\mu=0} = - [\Gamma_{gh, vac}^{(2)}(k_E^2) ]^{-1},
\end{align}
where $\Delta_{gh}$ is the ghost propagator.                                                                                                                                                                                                                                                                           

As a renormalization scheme, we adopt the ``infrared safe scheme''~\cite{TW11,PTW14} respecting the nonrenormalization theorem $ Z_A Z_C Z_{M^2} = 1$~\cite{Non-mass-ren}. For the gluon and ghost sector, we impose
\begin{align}
 \begin{cases}
 Z_A Z_C Z_{M^2} = 1 \\
 \Gamma_{{\mathscr A}, vac}^{(2)} (k_E^2 = \mu^2) = \mu^2 + M^2\\
 \Gamma_{gh, vac}^{(2)}(k_E^2 = \mu^2)= \mu^2 
 \end{cases}
\label{eq:TWrenomalization}
\end{align}
combined with the Taylor scheme \cite{Taylor71} $Z_g Z_A^{1/2} Z_C = 1$ for the coupling.\footnote{For the quark sector, we put $\Gamma^{(2)}_{s, vac} (k_E^2 = \mu^2) = m_q$ and $\Gamma^{(2)}_{v,vac}(k_E^2 = \mu^2) = 1$, where the quark propagator $\mathcal{S}(k_E) $ is parametrized as $\mathcal{S}^{-1}(k_E) = i \slashed{k}_E \Gamma_{v}^{(2)}(k_E)  + \Gamma_{s}^{(2)}(k_E)$. Note that this choice affects RG-improved results displayed in Fig.~\ref{fig:RG-winding}.}
In this renormalization scheme, it turns out that there exist RG flows on which the running coupling constant is always finite in a whole momentum region, which implies that the perturbation theory is valid to some extent.

\subsection{One-loop expressions}
Here we review the results of one-loop calculations of the in-medium gluon propagator. 

Beforehand, we decompose the vacuum polarization $\Pi_{\mu \nu} (k_E)$ into the vacuum part $\Pi^{vac}_{\mu \nu} (k_E)$ and the matter part $\Pi^{mat}_{\mu \nu} (k_E)$,
\begin{align}
\Pi_{\mu \nu} (k_E)
&= \Pi^{vac}_{\mu \nu} (k_E) + \Pi^{mat}_{\mu \nu} (k_E).
\end{align}
$\Pi^{vac}_{\mu \nu} (k_E)$ had been calculated in \cite{TW10, TW11, PTW14}. For completeness, the vacuum part is presented in Appendix B.

The relation between $\Pi_{\mu \nu} (k_E)$ and ${\mathscr D}_{\mu \nu} (k_E)$ is given by the further decomposition of $\Pi^{mat}_{\mu \nu} (k_E)$ as follows: in general, the spatial rotational symmetry yields
\begin{align}
\Pi^{mat}_{\mu \nu} (k_E) = \Pi_T ^{mat} (k_E^2) P^T_{\mu \nu} + \Pi_L ^{mat} (k_E^2) P^L_{\mu \nu} + \delta \Pi_{\mu \nu},
\end{align}
where the last term $\delta \Pi_{\mu \nu}$ is spanned by the tensorial structures $k_{E,\mu} k_{E,\nu}$ and $(P_{\mu \rho} t_\rho) k_{E,\nu} + (P_{\nu \rho} t_\rho) k_{E,\mu}$ with $t_\mu = (\vec{0},1)$ and does not contribute to the propagator due to the transversality of the Landau gauge, while the vacuum part can be written as 
$\Pi^{vac}_{\mu \nu} (k_E) = \Pi^{vac} (k_E^2) P_{\mu \nu}$.
The gluon propagator is thus of the form (\ref{eq:gluon_propagator_tensorial_str}) with the components of the vacuum polarization:
\begin{align}
{\mathscr D}_{\mu \nu} (k_E) &= 
D_T(k_E^2) P^T_{\mu \nu} + D_L(k_E^2) P^L_{\mu \nu}, \notag \\
D_T(k_E^2) &= \frac{1}{k_E^2 + \Pi^{vac} (k_E^2) + \Pi_T ^{mat} (k_E^2) }, \notag \\
D_L(k_E^2) &= \frac{1}{k_E^2 + \Pi^{vac} (k_E^2) + \Pi_L ^{mat} (k_E^2) }.
\end{align}

The matter part $\Pi^{mat}_{\mu \nu} (k_E)$ at zero temperature $T=0$ and nonzero quark chemical potential $\mu_q >0$ is the quark-loop contribution; for $\mu_q > m_q$, \cite{Kapusta-Gale}
\begin{align}
\Pi^{mat}_{\mu \nu}& (k_E) = \frac{1}{2} \left[ \Pi^{mat}_{\rho \rho} - \frac{k_E^2}{\vec{k}^2} \Pi^{mat}_{44} \right] P^T_{\mu \nu} + \frac{k_E^2}{\vec{k}^2} \Pi^{mat}_{44} P^L_{\mu \nu}, \\
\Pi^{mat}_{\rho \rho} &= 2 \frac{g^2 C(r)}{\pi^2} \operatorname{Re} \int_0^{p_F} \frac{dp p^2}{E_p} \left[ 1- \frac{2m_q^2 - k_E^2}{4 p |\vec{k}|} \ln \left( \frac{R_+}{R_-} \right) \right], \notag \\
\Pi^{mat}_{44} &= \frac{g^2 C(r)}{\pi^2} \operatorname{Re} \int_0^{p_F} \frac{dp p^2}{E_p} \notag \\
&~~~\left[ 1 - \frac{k_E^2+4 E_p^2 + 4 i E_p k_4}{4 p |\vec{k}|} \ln \left( \frac{R_+}{R_-} \right) \right],
\end{align}
where $C(r) = N_F/2,~ p_F = \sqrt{\mu_q^2 - m_q^2},~ E_p =\sqrt{p^2 + m_q^2}$, 
\begin{align}
R_\pm = -k_E^2 + 2i k_4 E_p \pm 2 p |\vec{k}|,
\end{align}
and $\operatorname{Re}$ denotes the real part when $k_4$ is real, namely, $\operatorname{Re} f(i k_4) := \frac{1}{2} (f(ik_4) + f(-ik_4))$ for any function $f(ik_4)$.

Now, since we are interested in complex mass and long-distance behavior, let us take the long-wavelength limit $\vec{k} \rightarrow 0$ symmetrically. This limit reduces technical difficulties on the analytic continuation significantly.

In the long-wavelength limit $\vec{k} \rightarrow 0$, we have
\begin{align}
P_{ij} &= \delta_{ij},~~ P_{4i} = P_{i4} = P_{44} = 0, ~~(i,j = 1,2,3) \notag \\
P_{\mu \nu}^T &\rightarrow \frac{2}{3} P_{\mu \nu}, ~~ P_{\mu \nu}^L \rightarrow \frac{1}{3} P_{\mu \nu}, \notag \\
& \Pi^{mat}_{\mu \nu}  (k_E) = \frac{1}{3} \Pi^{mat}_{\rho \rho} P_{\mu \nu},
\end{align}
and,
\begin{align}
&\Pi^{mat}_{\mu \mu} (\vec{k} \rightarrow 0, k_4) = \frac{g^2 C(r)} {4 \pi ^2 k_4} \theta (\mu_q - m_q)  \Bigl[ 4 k_4 p_F \sqrt{p_F^2+m_q^2} \notag \\
&+2 k_4^3 \ln
   \left(\frac{m_q}{\sqrt{p_F^2+m_q^2}+p_F}\right)+\left(2 m_q^2-k_4^2\right) \sqrt{k_4^2+4 m_q^2}  \notag \\
& \times   \ln \left(\frac{\sqrt{\left(k_4^2+4 m_q^2\right)
   \left(p_F^2+m_q^2\right)}-k_4 p_F}{\sqrt{\left(k_4^2+4 m_q^2\right)
   \left(p_F^2+m_q^2\right)}+k_4 p_F}\right)\Bigr],
\end{align}
where $\theta (\mu_q - m_q)$ is the step function.
Then, the gluon propagator ${\mathscr D}_{\mu \nu} (\vec{k} \rightarrow 0, k_4)$ can be written as
\begin{align}
{\mathscr D}_{\mu \nu} (k_4) &= {\mathscr D}_T (-k_4^2) P_{\mu \nu}, \notag \\
{\mathscr D}_T (-k_4^2) &= \frac{1}{M^2(s+1+\hat{\Pi}^{vac}(s)+\hat{\Pi}^{mat}(s))},
\label{eq:gluon-propagator}
\end{align}
where
\begin{align}
s = \frac{k_4^2}{M^2},
\end{align}
$\hat{\Pi}^{vac}(s)$ is the vacuum part given in Appendix B (\ref{eq:vacuum_pol_TW}), and,
\begin{align}
&\hat{\Pi}^{mat}(s) =  \frac{g^2 C(r)} {12 \pi ^2} \theta (\zeta -\xi) \Bigl[ 4 \sqrt{\zeta (\zeta - \xi)} \notag \\
&+2 s \ln
   \left(\frac{\sqrt{\xi}}{\sqrt{\zeta}+\sqrt{\zeta - \xi}}\right)
+\frac{1}{\sqrt{s}}\left(2 \xi - s \right) \sqrt{s+4 \xi}  \notag \\
& \times   \ln \left(\frac{\sqrt{\zeta (s + 4 \xi)}- \sqrt{s (\zeta -\xi)}}{\sqrt{\zeta (s + 4 \xi)} + \sqrt{s (\zeta -\xi)}}\right)\Bigr].
\end{align}
with
\begin{align}
 \xi = \frac{m_q^2}{M^2}, ~ \zeta = \frac{\mu_q^2}{M^2}.  \label{eq:xi-zeta}
\end{align}

Notice that
\begin{align}
&\hat{\Pi}^{mat}(s \rightarrow 0) =  \frac{g^2 C(r)} {3 \pi ^2} \theta (\zeta -\xi) \Bigl[  \frac{(\zeta - \xi)^{3/2}}{\sqrt{\zeta}}\Bigr] > 0, \label{IR_matter}
\end{align}
and
\begin{align}
&\hat{\Pi}^{mat}(s \rightarrow \infty) = O(s). \label{UV_matter}
\end{align}

\section{Results}

In this section, we study the analytic structure of the gluon propagator with the one-loop quantum corrections presented in the previous section.

From here on, we set $G = SU(3)$ and the renormalization scale $\mu_0 = 1 ~\mathrm{GeV}$. With the RG improvements, the best-fit parameters reported in \cite{PTW14} are
\begin{subequations}
\begin{align}
    g = 4.5,~ M = 0.42 ~\mathrm{GeV}, \label{eq:PTW_parameter}
\end{align}
and the up and down quark mass parameters
\begin{align}
    m_q = 0.13 ~\mathrm{GeV}, \label{eq:PTW_parameter_mq}
\end{align}
in the case of $N_F = 2$.
\end{subequations}





\subsection{Number of complex poles}

 \begin{figure}[tb]
  \begin{center}
   \includegraphics[width=0.9\linewidth]{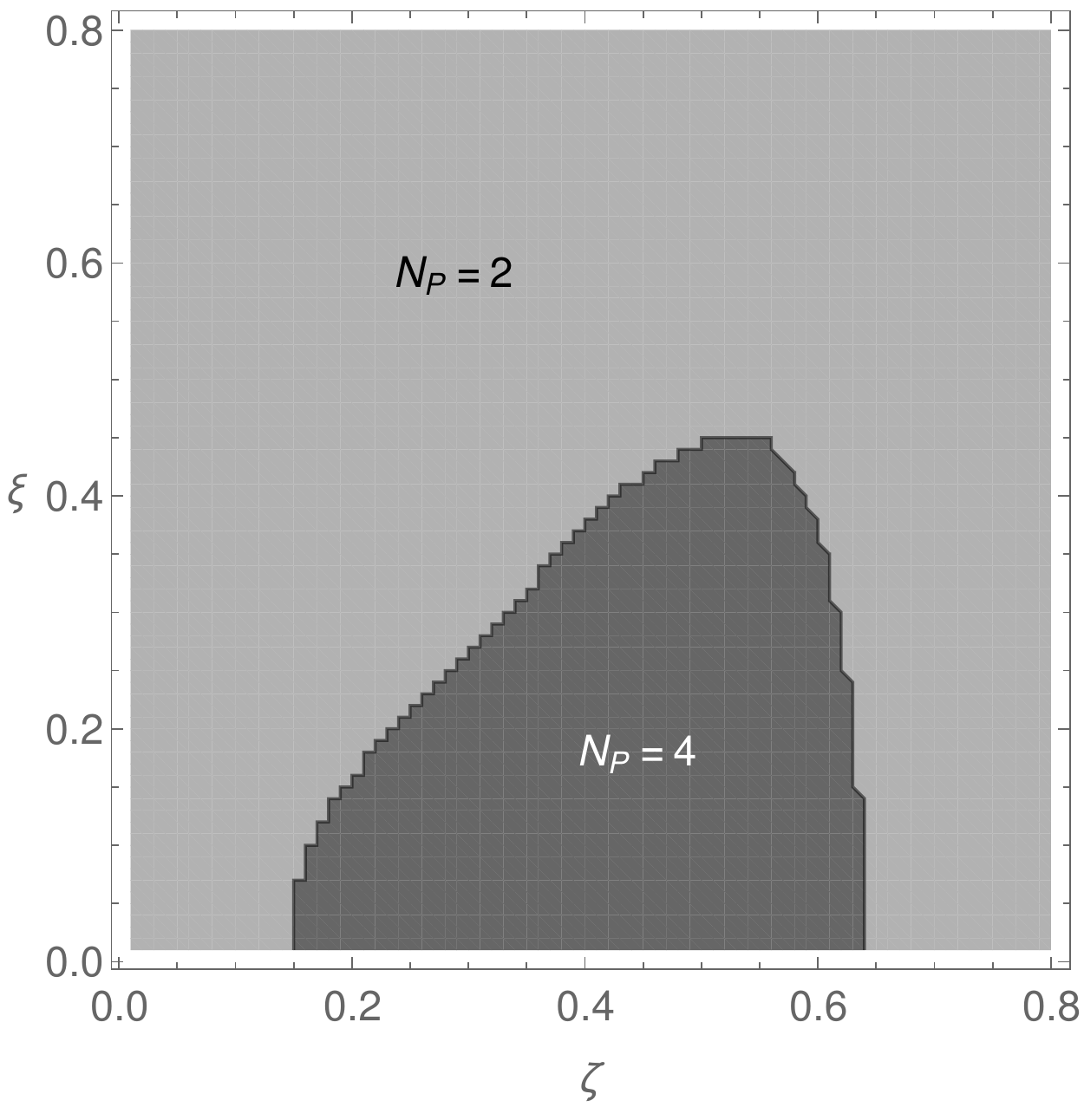}
  \end{center}
   \caption{Contour plot of $N_W(C)$ for the gluon propagator on the $(\zeta = \frac{\mu_q^2}{M^2},\xi= \frac{m_q^2}{M^2})$ plane at the set of parameters (\ref{eq:PTW_parameter}), which gives the number of complex poles through the relation $N_P = - N_W(C)$. In the $N_P = 2,4$ regions, the gluon propagator has one pair and two pairs of complex conjugate poles, respectively. We used $N = 8 \times 10^5$, $x_n = (n+1) \times 10^{-5} M^2~ (n = 0, \cdots, N)$, and $x_{N+1} = 50 M^2$ for the discretization (\ref{eq:winding-timelike}) and $\epsilon = 10^{-9} M^2$ for the infinitesimal imaginary part. For larger $\zeta$ or $\xi$, the gluon propagator has one pair of complex conjugate poles.}
    \label{fig:typical_windingnumber}
\end{figure}

First, we compute the number of complex poles for the one-loop gluon propagator (\ref{eq:gluon-propagator}) at the parameters (\ref{eq:PTW_parameter}), and $N_F = 2$ by using the winding number $N_W(C)$ defined in (\ref{eq:winding-timelike}) of Sec.~IIB.
We analytically continue the gluon propagator ${\mathscr D}_T (-k_4^2)$ from the Euclidean axis $z^2 = - k_4^2$ to the whole $z^2$ plane. In terms of (\ref{eq:generalized-spec-repr}),
\begin{align}
    D(z,\vec{k} \rightarrow 0) = {\mathscr D}_T (z^2).
\end{align}
Let us check the prerequisites for the claim of Sec.~IIB.
The gluon propagator takes the form (\ref{eq:generalized-spec-repr}), since it has no branch cut except for the real axis as can be confirmed from (\ref{eq:gluon-propagator}). Thus, it can have only complex poles in the complex plane excluding the real axis. 
Also, this gluon propagator satisfies the conditions (i) and (ii) in Sec.~II~B and has no zeros $N_Z = 0$:
\begin{itemize}
    \item As $|z^2| \rightarrow \infty$, 
\begin{align}
\mathscr{D}_T (z^2) \simeq [g^2 \gamma_0 (-z^2) \ln |z^2| + O(z^2)]^{-1}, 
\end{align}
    from (\ref{eq:gluon_asymptotic_UV_one_loop}) and (\ref{UV_matter}) as desired.\footnote{Although the naive one-loop asymptotic form has the wrong exponent of the logarithm $(\ln |z^2|)$, we can expect this does not change $N_W(C)$ as it has similar phase to the correct one (for $N_F < 10$). See Appendix B.}
    \item As $|z^2| \rightarrow 0$, 
\begin{align}
\mathscr{D}_T (z^2) > 0,
\end{align}
from (\ref{eq:pos_IR_prop_cond}) and (\ref{IR_matter}) as desired.
    \item The gluon propagator has no zeros $N_Z = 0$, since the inverse of the propagator (\ref{eq:gluon-propagator}) does not diverge.
\end{itemize}
Therefore, the number of complex poles can be calculated according to (\ref{eq:winding-timelike}) and $N_P = - N_W(C)$. For the condition (iii) in Sec.~II~B, we numerically check convergence of the refinement of the discretization.

Figure \ref{fig:typical_windingnumber} is a contour plot of $N_W(C)$ on the plane $(\mu_q^2, m_q^2)$ normalized by the gluon mass $M^2$, i.e. $(\zeta, \xi )$ plane with (\ref{eq:xi-zeta}). At the vacuum case $\mu_q = 0$, the gluon propagator has one pair of complex conjugate poles, namely two complex poles ($N_P = 2$), irrespective of the value $m_q$. The novel $N_P = 4$ region appears for light quarks ($\xi \lesssim 0.5$, or $m_q \lesssim 0.30$ GeV). As the quark chemical potential $\mu_q$ increases for such light quarks, the number of complex poles becomes four ($N_P = 4$) at slightly above the quark mass $m_q$ and backs to two ($N_P = 2$) at $\zeta \approx 0.6$, or $\mu_q \approx 0.33$ GeV. In the intermediate quark chemical potential, the gluon propagator has four complex poles in complex $z^2$ plane. For large $m_q$ or $\mu_q$, the gluon propagator has two complex poles as in the vacuum case.

\subsection{Analytic structure at a specific set of parameters}

 \begin{figure}[tb]
 \begin{minipage}{\hsize}
  \begin{center}
   \includegraphics[width= \linewidth]{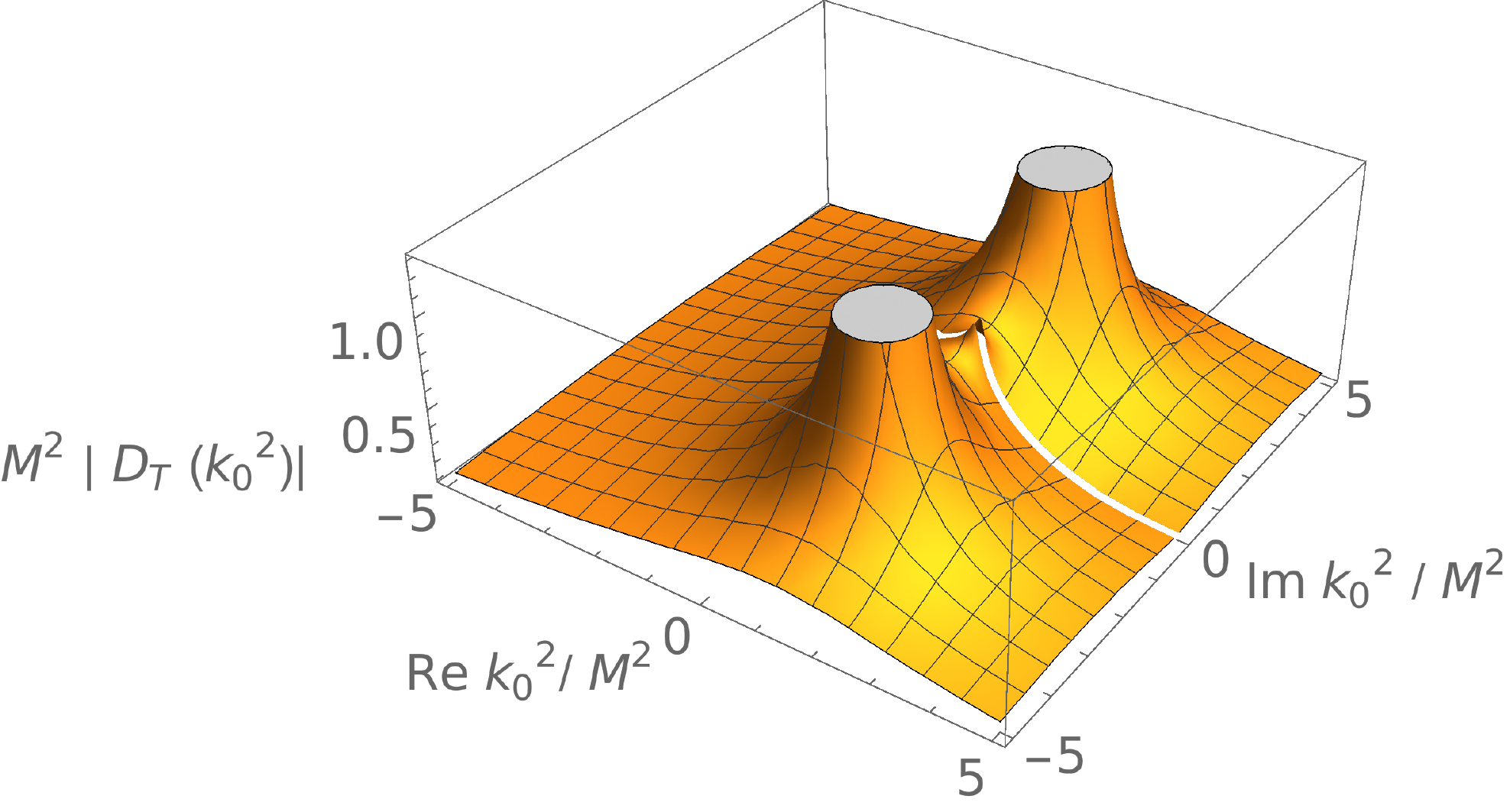}
  \end{center}

 \end{minipage}
 \begin{minipage}{\hsize}
  \begin{center}
   \includegraphics[width= \linewidth]{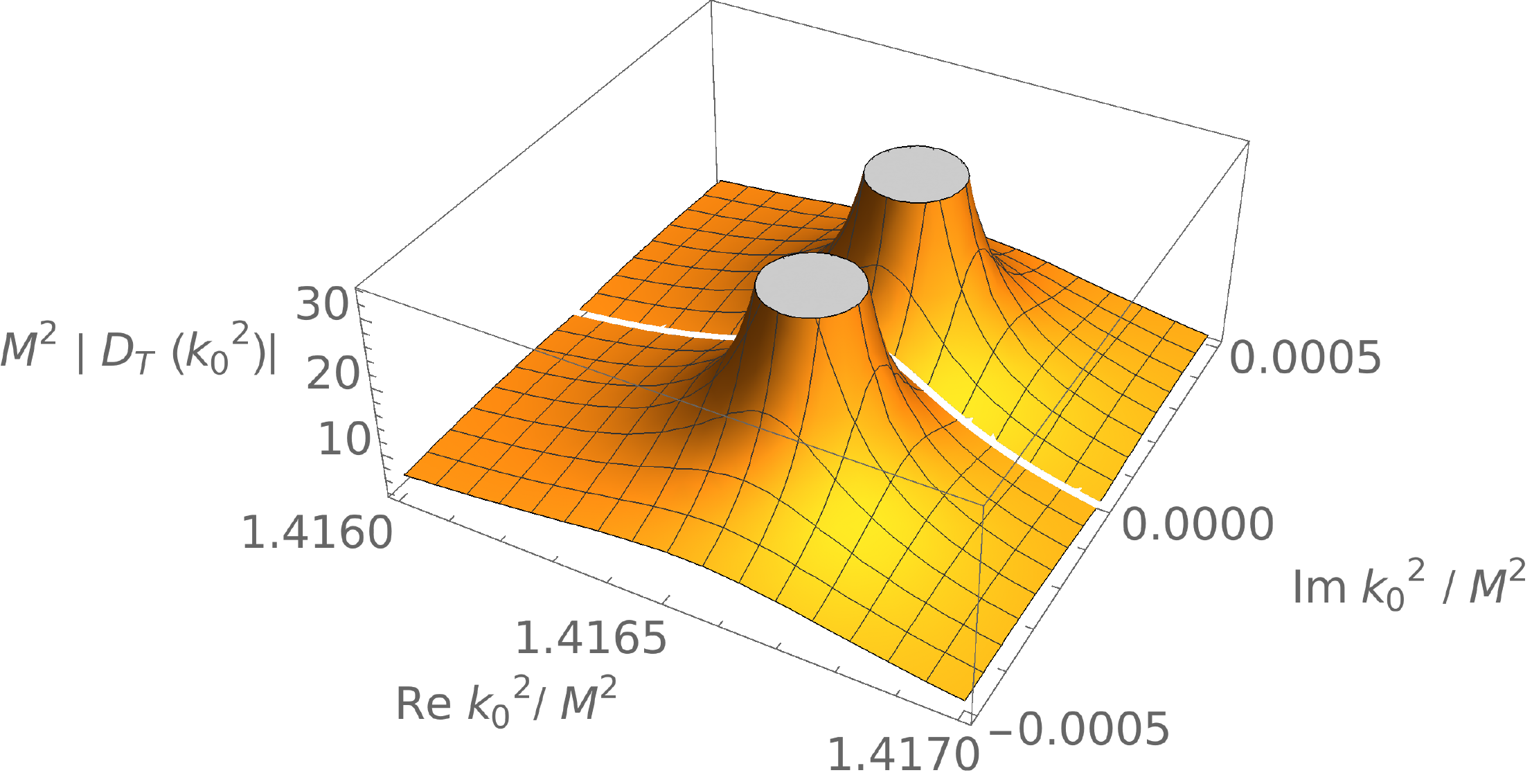}
  \end{center}
 \end{minipage}
 \caption{ Modulus of the gluon propagator $|\mathscr{D}_T(k_0^2)|$ with the set of parameters (\ref{eq:parameter_specific}) on the complex $k_0^2$ plane. The top panel is written in the range of $-5 < \operatorname{Re} k_0^2/ M^2 <5, ~ -5 < \operatorname{Im} k_0^2/ M^2 <5$. 
 A pair of complex conjugate poles is clearly illustrated.
 The other pair of complex conjugate poles exists at $\operatorname{Re} k_0^2/M^2 \approx 1.4$ near the real axis of $k_0^2$, which is however difficult to be identified in the top panel, and hence is enlarged to be visible in the range $1.416 < \operatorname{Re} k_0^2/ M^2 < 1.417, ~ - 5 \times 10^{-4} < \operatorname{Im} k_0^2/ M^2 < 5 \times 10^{-4}$ in the bottom panel.
 }
    \label{fig:025prop_3D}
\end{figure}

 \begin{figure}[tb]
 \begin{minipage}{\hsize}
  \begin{center}
   \includegraphics[width=0.9\linewidth]{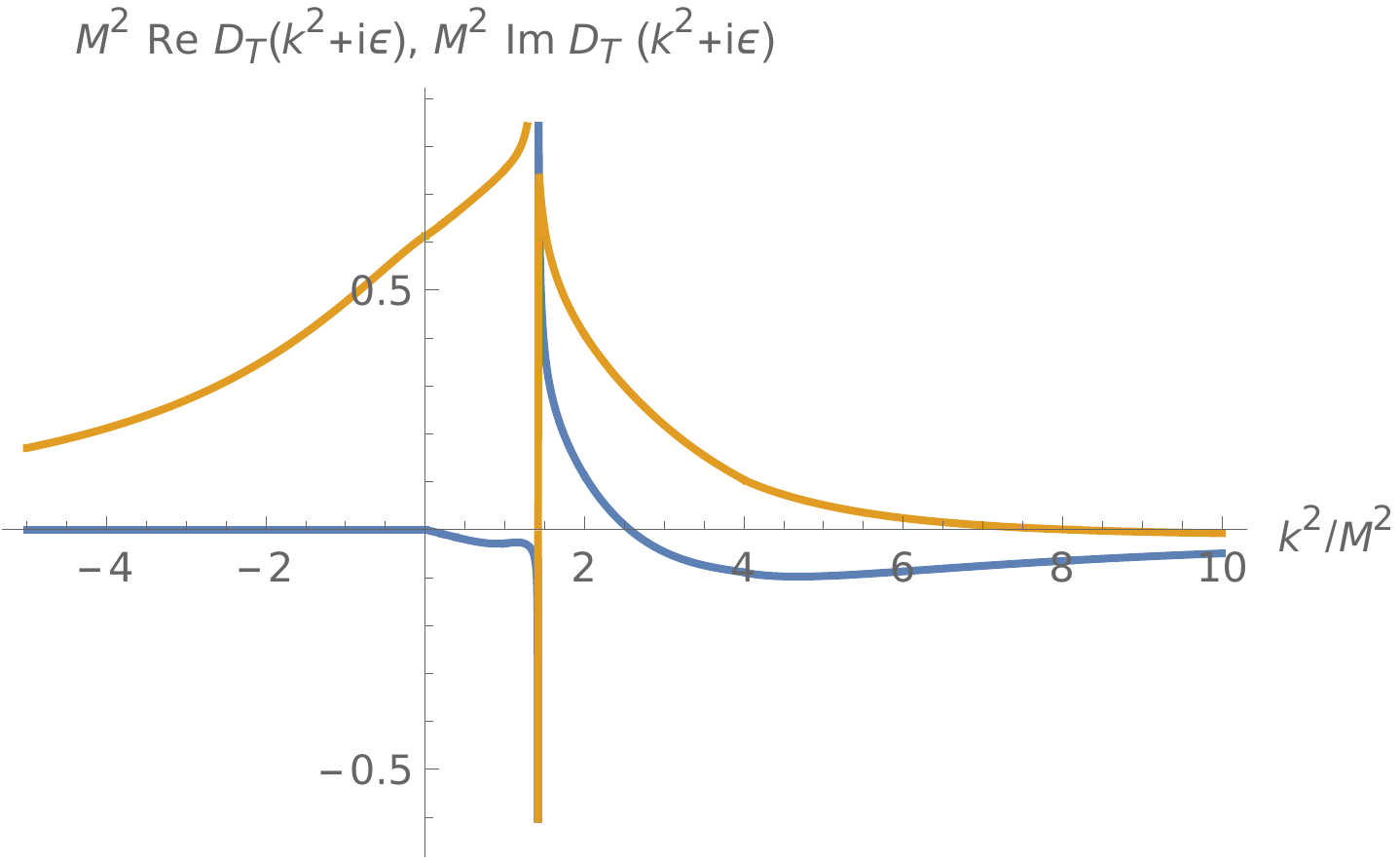}
  \end{center}

 \end{minipage}
 \begin{minipage}{\hsize}
  \begin{center}
   \includegraphics[width=0.9\linewidth]{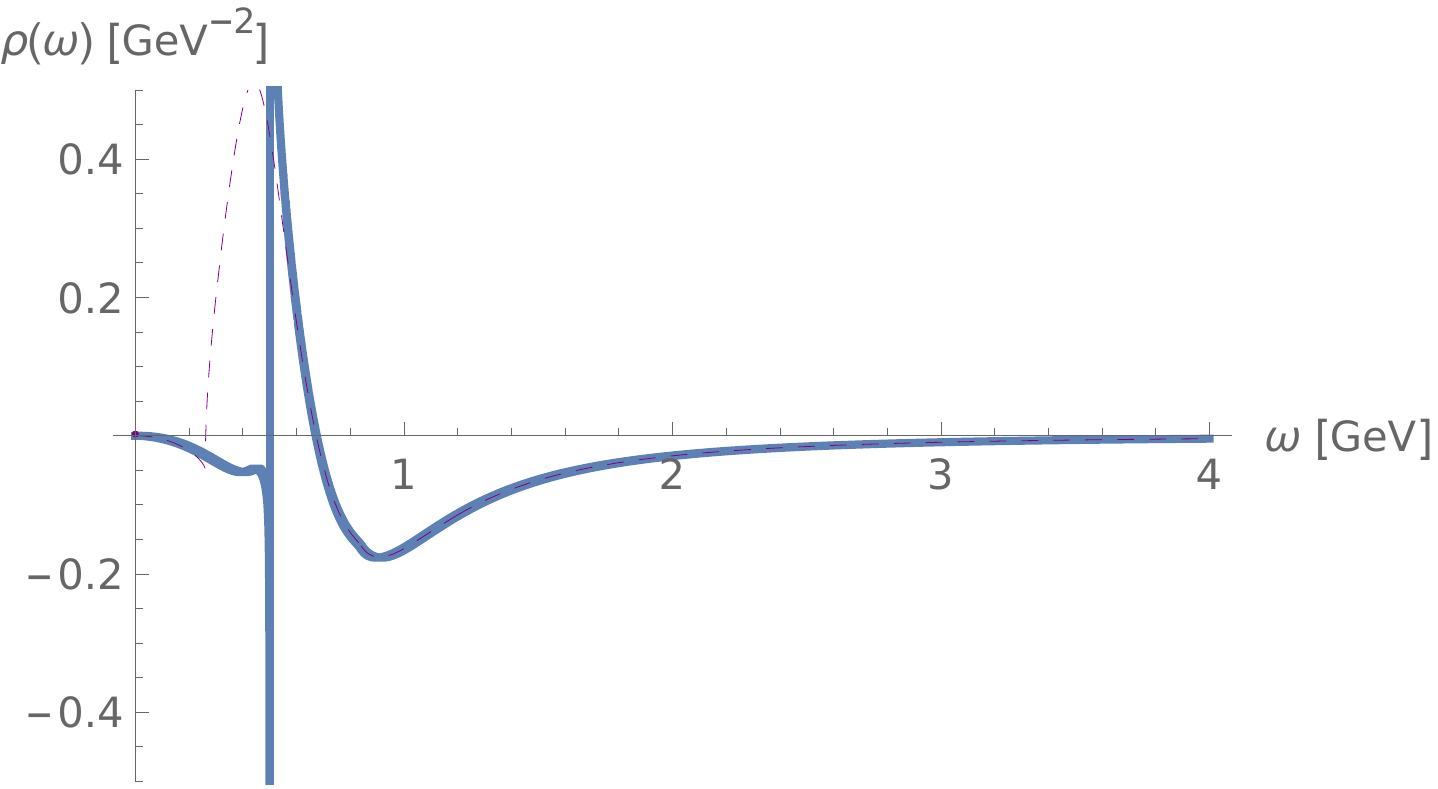}
  \end{center}
 \end{minipage}
 \caption{Top panel: real (orange) and Imaginary (blue) parts of the gluon propagator (\ref{eq:gluon-propagator}) with real $k_0^2$ at the set of parameters (\ref{eq:parameter_specific}) and $N_F =2$. The peak at $k_0^2/M^2 \approx 1.4$ reflects the fact that the gluon propagator at this set of parameters has a pair of almost real complex poles.
 Bottom panel: the spectral function at the same set of parameters. A pair of positive and negative peaks is located at $\omega \approx 0.5$ GeV. At $\omega \approx 0.5$, the positive peak lasts up to $\max{\rho} \sim 2.7~ \mathrm{GeV^{-2}}$ and the negative one to $\min{\rho} \sim -29~ \mathrm{GeV^{-2}}$. The purple dashed curve plots the vacuum one $\mu_q = 0$. In the $\omega \rightarrow 0$ and $\omega \rightarrow \infty$ limit, both of them exhibit similar behavior.}
    \label{fig:025prop_spec}
\end{figure}

Next, we take a further look at the analytic structure of the gluon propagator at a specific set of parameters. As the $N_P = 4$ region with intermediate $\mu_q$ will be interesting, let us choose (\ref{eq:PTW_parameter}), (\ref{eq:PTW_parameter_mq}), $\mu_q = 0.25 \textrm{ GeV}$, i.e.,
\begin{align}
    (g,M,m_q , \mu_q) = (4.5,0.42 \textrm{ GeV}, 0.13 \textrm{ GeV}, 0.25 \textrm{ GeV}), \label{eq:parameter_specific}
\end{align}
and $N_F =2$.

In what follows, we use $k_0$ to denote the complex variable $z$:
\begin{align}
k_0 := z.
\end{align}
To take a look at the analytic structure of the gluon propagator, let us see its modulus on the complex $k_0^2$ plane. The modulus of the gluon propagator $\mathscr{D}_T (k_0^2)$ is plotted in Fig.~\ref{fig:025prop_3D}. We can observe that the gluon propagator at the given parameters (\ref{eq:parameter_specific}) has indeed two pairs of complex conjugate poles. One pair that is clearly visible in the top panel of Fig.~\ref{fig:025prop_3D} is located at $k_0^2/M^2 \approx 1.4 \pm 2.6 i$, or $k_0 \approx \pm 0.62 \pm 0.36 i$ GeV. The other pair of complex conjugate poles is at  $k_0^2/M^2 \approx 1.4 \pm (1.8 \times 10^{-4}) i$, or $k_0 \approx \pm 0.50 \pm (3.1 \times 10^{-5})i$ GeV. 

The latter pair has very small imaginary part, while the former one is similar to that in the vacuum case. This smallness of the imaginary part is a universal feature not only around the transition, but also on the whole $N_P = 4$ region, as we will see in the next subsection.

The gluon propagator (\ref{eq:gluon-propagator}) with real $k_0^2$ and its spectral function,
\begin{align}
    \rho(\omega) := \rho (\omega,\vec{k} \rightarrow 0) := \frac{1}{\pi} \operatorname{Im} \mathscr{D}_T (\omega^2 + i \epsilon),
\end{align}
are displayed in Fig.~\ref{fig:025prop_spec}. The propagator shows a rapid oscillation at $k_0^2/M^2 = - k_4^2/M^2 \approx 1.4$, or $k_0 \approx 0.5 \textrm{~GeV} \approx 2 \mu_q$. The negative peak of the spectral function has a larger value than the positive one: $\max{\rho} \sim 2.7~ \mathrm{GeV^{-2}}$ and $\min{\rho} \sim -29~ \mathrm{GeV^{-2}}$.
The rapid change is consistent with existence of almost real complex poles.
Apart from the sharp peak, the gluon propagator is similar to the vacuum one. The quark chemical potential affects the gluon propagator significantly only around $k_0 \approx 2 \mu_q$.

\subsection{Locations of complex poles}

Let us investigate locations of complex poles of the gluon propagator for various parameters $(\zeta = \frac{\mu_q^2}{M^2}, \xi = \frac{m_q^2}{M^2} )$ with fixed $(g,M)$ of (\ref{eq:PTW_parameter}). We present the ratio $\omega_I/ \omega_R$ of the real and imaginary parts of a complex pole,
\begin{align}
    k_0 = \omega_R + i \omega_I \in \mathbb{C},
\end{align}
on the $(\zeta,\xi)$ plane and a trajectory of poles for varying $\mu_q$ and at fixed $m_q$.

First, we compute the ratio $\omega_I/ \omega_R$ to obtain an overview on positions of complex poles on the parameter space $(\zeta, \xi)$. We can restrict ourselves to $\omega_R >0,~\omega_I > 0$ without loss of generality from the Schwarz reflection principle and the symmetry $k_0 \rightarrow - k_0$. As the gluon propagator has at most two pairs of complex conjugate poles with respect to $k_0^2$, it is sufficient to find $\max \omega_I/ \omega_R$ and $\min \omega_I/ \omega_R$.

 \begin{figure}[tb]
 \begin{minipage}{\hsize}
  \begin{center}
   \includegraphics[width=\linewidth]{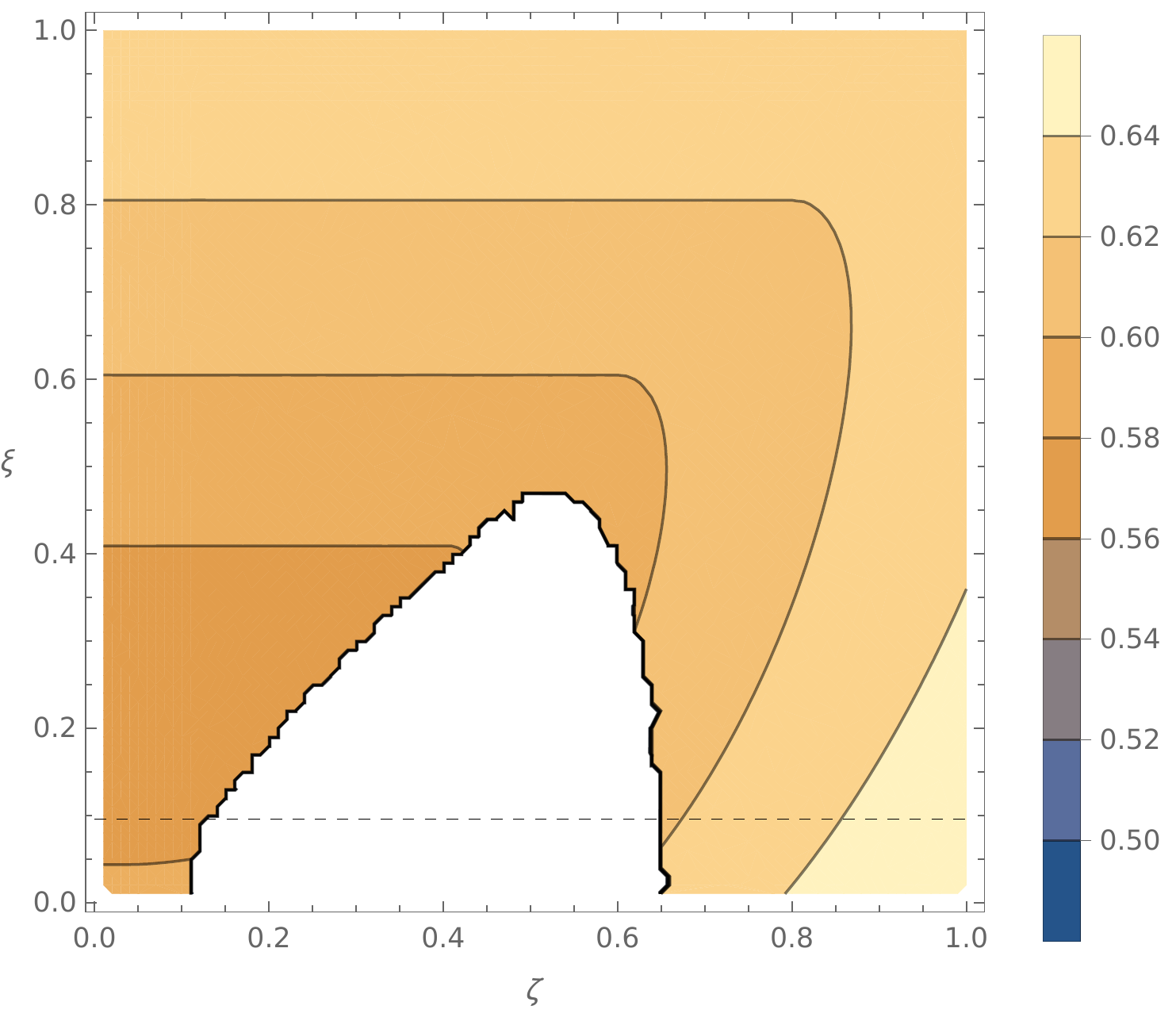}
  \end{center}

 \end{minipage}
 \begin{minipage}{\hsize}
  \begin{center}
   \includegraphics[width=\linewidth]{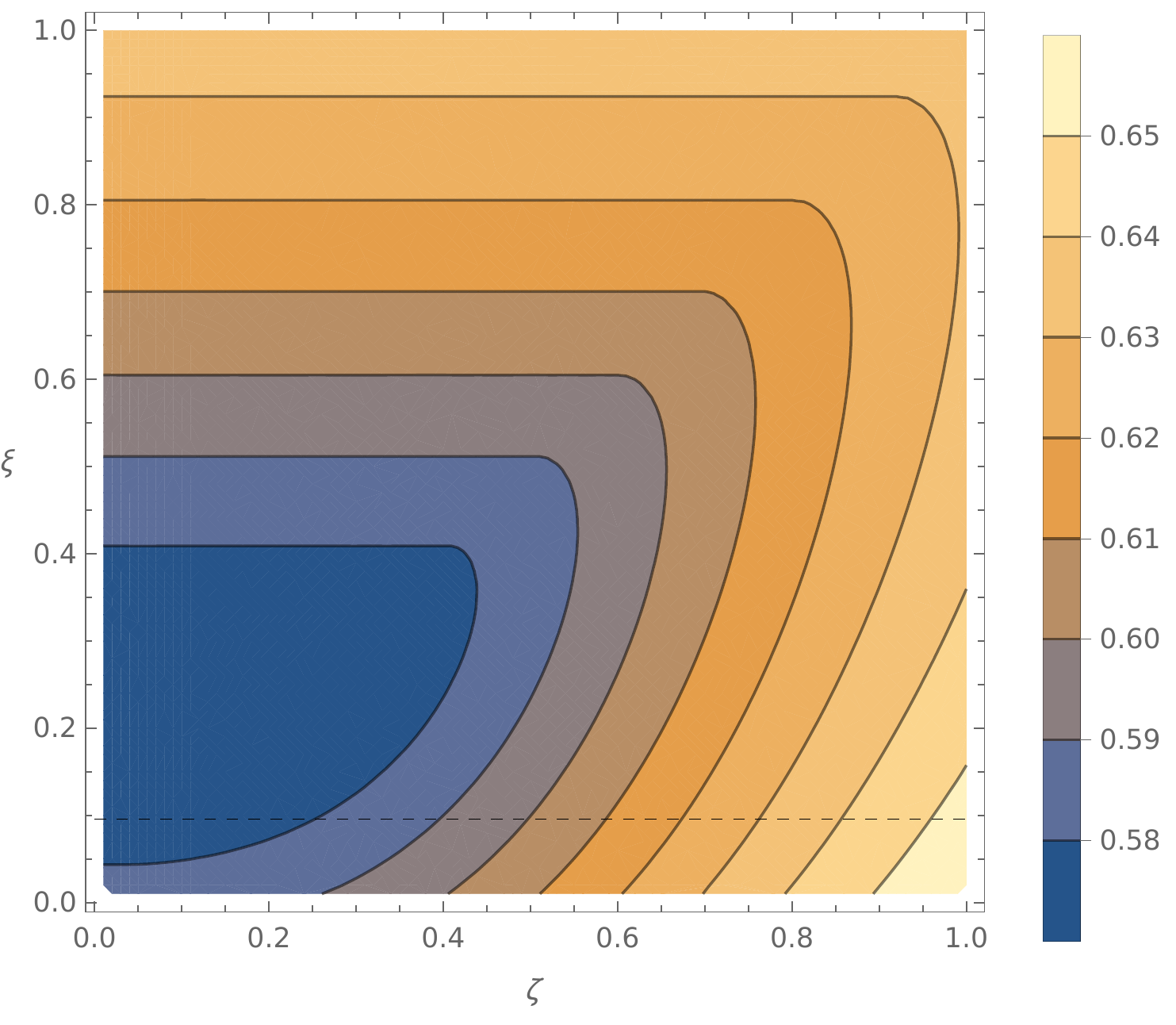}
  \end{center}
 \end{minipage}
 \caption{Contour plots of $\min \omega_I/ \omega_R$ (top) and $\max \omega_I/ \omega_R$ (bottom) for a complex pole $k_0 = \omega_R + i \omega_I, ~ (\omega_R >0,~\omega_I > 0)$ of the gluon propagator on the $(\zeta = \frac{\mu_q^2}{M^2}, \xi = \frac{m_q^2}{M^2})$ plane. 
The region of $\min \omega_I/ \omega_R < 10^{-3}$ is represented by a blank, where the gluon propagator has two pairs of complex conjugate poles. The horizontal dashed line is at $\xi = 0.096$, or $m_q = 0.13$ GeV.}
    \label{fig:ratio_min_max}
\end{figure}

Contour plots of the ratios ($\max \omega_I/ \omega_R$ and $\min \omega_I/ \omega_R$) are shown in Fig.~\ref{fig:ratio_min_max}. 
This result is consistent with Fig.~\ref{fig:typical_windingnumber} as $\max \omega_I/ \omega_R \neq \min \omega_I/ \omega_R$ only on the $N_P = 4$ region, where the gluon propagator $\mathscr{D}_T(k_0^2)$ has two pairs of complex conjugate poles with respect to $k_0^2$.
These figures indicate that the gluon propagator has a pair of almost real complex poles in the $N_P = 4$ region shown, while the pair with $\max \omega_I/ \omega_R$ is always of the same order of magnitude.

Moreover, in general, the ratio $\omega_I/ \omega_R$ tends to increase as the quark chemical potential $\mu_q$ increases, except for the almost real poles.
In other words, the gluon propagator becomes ``less particlelike'' for large $\mu_q$.

In the previous subsection, we observed that both the sharp spectral peak and almost real poles appear at $\operatorname{Re} k_0 \approx 2 \mu_q ( \approx 0.5 \textrm{~GeV})$ at $\mu_q = 0.25$ GeV. This feature is not limited to the specific parameter but universal.
Let us examine locations of complex poles at the parameter (\ref{eq:PTW_parameter}) and (\ref{eq:PTW_parameter_mq}) with varying $\mu_q$.

 \begin{figure}[tb]
 \begin{minipage}{\hsize}
  \begin{center}
   \includegraphics[width=0.9\linewidth]{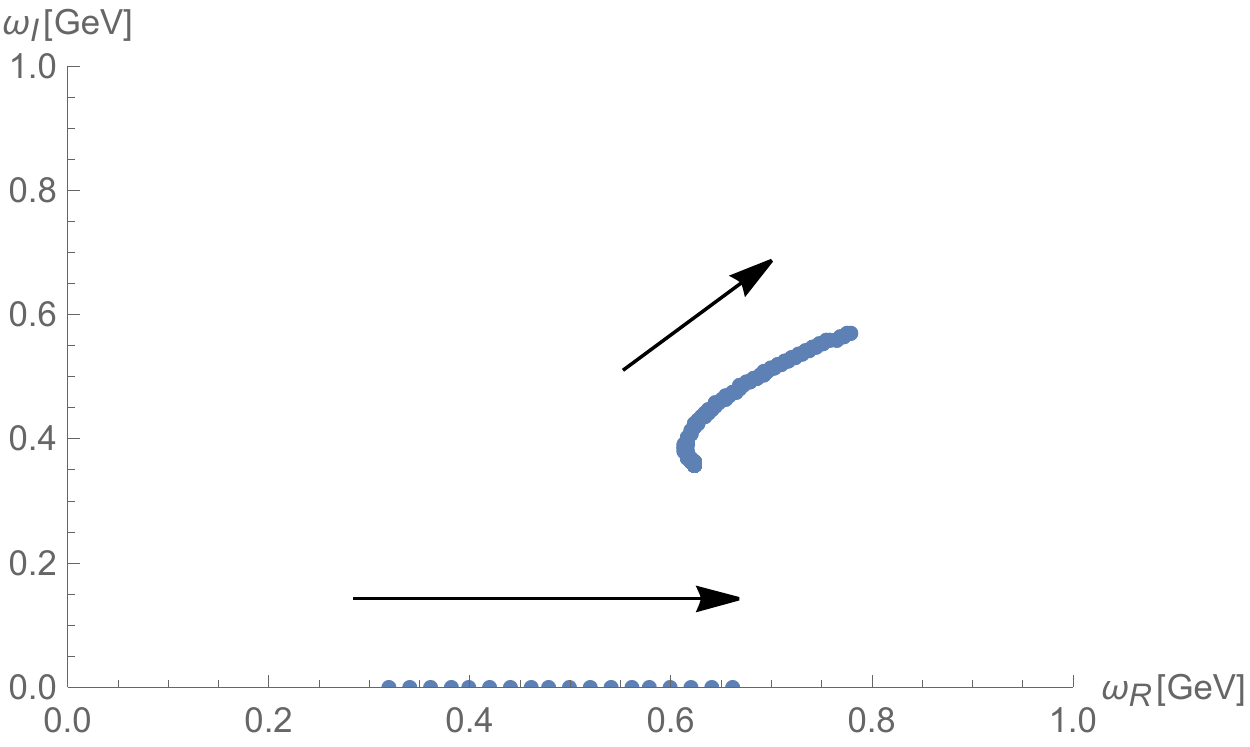}
  \end{center}

 \end{minipage}
 \begin{minipage}{\hsize}
  \begin{center}
   \includegraphics[width=0.9\linewidth]{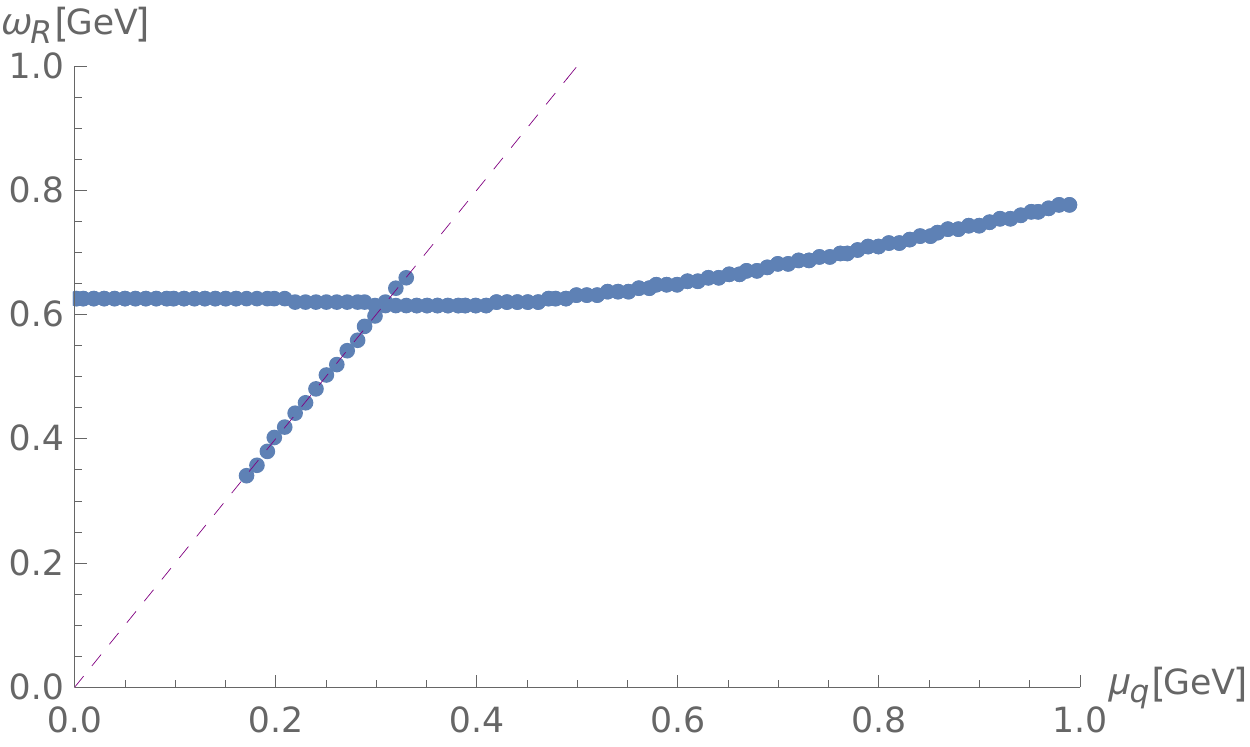}
  \end{center}
 \end{minipage}
 \caption{Top panel: trajectory of a complex pole $k_0 = \omega_R + i \omega_I, ~ (\omega_R >0,~\omega_I > 0)$ of the gluon propagator in the plane $(\omega_R,\omega_I)$ at the parameter (\ref{eq:PTW_parameter}) and (\ref{eq:PTW_parameter_mq}) with varying $\mu_q$ from 0 to 1 GeV. As $\mu_q$ increases, the poles move along the arrows. Note that the almost real pole ($\omega_I \approx 0$) exists only for the $N_P = 4$ region while the other pole for any value of $\mu_q$.
 Bottom panel: $\mu_q$ dependence of the real part of location of a complex pole.
 The data of the new complex poles are approximated by the straight line $\omega_R = 2 \mu_q$ (purple dashed line) well. This figure shows the almost real pole appears at $\mu_q \approx 0.16$ GeV and disappears at $\mu_q \approx 0.33$ GeV in agreement with Fig.~\ref{fig:typical_windingnumber}.}
    \label{fig:pole_traj}
\end{figure}

Figure \ref{fig:pole_traj} plots a trajectory of complex poles on the complex $k_0$ plane and $\mu_q$-dependence of the real parts of the complex poles.
As $\mu_q$ increases, a new pole appears from the branch cut (at $\mu_q \approx 0.16$ GeV), then moves along the real axis, and is finally absorbed into the branch cut (at $\mu_q \approx 0.33$ GeV). On the other hand, the other pole increases its imaginary part gradually. This feature is consistent with the number of complex poles of Sec. IV A.

The bottom panel of Fig.~\ref{fig:pole_traj} clearly indicates that the real part of the new almost real pole can be approximated by $2 \mu_q$: $\omega_R \approx 2 \mu_q$.
We have also checked that the almost real poles are at $\operatorname{Re} k_0 \approx 2 \mu_q$ for different values of $m_q$.

\subsection{$(g,M)$ dependence}

Before concluding this section, let us consider $(g,M)$ dependence of the above results, especially, the number of complex poles. For details of these analyses, see Appendix C. We have found that the $N_P$ contour plot is not sensitive to a detailed choice of the parameters $(g,M)$.

\subsection{Summary of results}

In summary, we have observed the following points in this section.
\begin{itemize}
    \item There is a $N_P = 4$ region, where the gluon propagator has two pairs of complex conjugate poles with respect to $k_0^2$. See Fig.~\ref{fig:typical_windingnumber}.
    \item  In the $N_P = 4$ region, the gluon propagator has an almost real pair of complex conjugate poles at $\operatorname{Re} k_0 \approx 2 \mu_q$. See Fig.~\ref{fig:pole_traj}
    \item With almost real poles, the real part and imaginary part (to be identified with the spectral function) have narrow peaks at $k_0 \approx 2 \mu_q$. See Fig.~\ref{fig:025prop_spec}
    \item The ratio $\omega_I/ \omega_R$ of a complex pole $k_0 = \omega_R + i \omega_I, ~ (\omega_R >0,~\omega_I > 0)$ tends to increase as $\mu_q$ increases, except for the almost real poles. See Fig.~\ref{fig:ratio_min_max}.
\end{itemize}

\section{Discussion}

In this section, we first discuss implications of the results shown in the previous sections, especially the appearance of the almost real pole in the $N_P = 4$ region. Second, we compute $N_W(C)$ of the RG-improved gluon propagator to examine the robustness.
Third, we attempt to estimate the analytic structure of the gluon propagator for relatively large $\mu_q$. Finally, we comment on the infrared problems which appear in the thermal context.

\subsection{Almost real complex poles and spectral function}
For the gluon propagator, we found a new pair of complex conjugate poles at $\operatorname{Re} k_0 \approx 2 \mu_q$ with quite small imaginary parts ($\operatorname{Im} k_0 \approx 0$). Together with the narrow peaks shown in Fig.~\ref{fig:025prop_spec}, the quark chemical potential affects the gluon propagator significantly around $k_0 \approx 2 \mu_q$. 

The importance of the scale $2 \mu_q$ can be understood by the fact that $2 \mu_q$ is the lowest energy for the quark pair creation to occur, which contributes to the spectrum of the gluon, due to the Fermi degeneracy.
Moreover, in the massive model, the gluon ``decouples'' at low energies. Thus, quark loop dominates the low-energy region of the gluon spectral function.
On the other hand, in the high energy region, the gluon and ghost loops win against the quark loop for the gluon spectral function to yield $\rho < 0$ in the large frequency limit for $N_F < 10$ \cite{OZ80}. Therefore, $2 \mu_q$ will be quite an important scale for relatively small $\mu_q$ (but larger than $m_q$), while less important in the high-energy region.
This might explain the appearance and disappearance of the almost real complex poles as varying $\mu_q$.

Since complex poles never appear in the physical spectrum, they should correspond to confined degrees of freedom.
The transition between the $N_P = 2$ and $N_P = 4$ regions indicates that timelike spectra transform to confined complex degrees of freedom, or vice versa.
Therefore, the transition between the $N_P = 2$ and $N_P = 4$ regions might have a physical significance on the dynamics of the strong interaction.

Note that, however, the appearance of the almost real pole may be an artifact of the approximation:
\begin{align}
{\mathscr D}_T (-k_4^2) &\approx \frac{1}{k_4^2+M^2+\Pi^{1-loop}(k_4^2)}, \label{eq:approx_oneloop}
\end{align}
where the vacuum polarization $\Pi$ is replaced by the one-loop expression $\Pi^{1-loop}$.
For example, in this approximation, even the propagator of the Higgs field in $U(1)$ Higgs model with the small gauge-fixing parameter has complex poles with tiny imaginary parts \cite{DEGHMPPS19}. 
The new pole reported in the previous section may be similar to this one.
In this case, the almost real pole should be interpreted as a long-lived collective mode with frequency $2 \mu_q$.

 \begin{figure}[t]
  \begin{center}
   \includegraphics[width=0.9\linewidth]{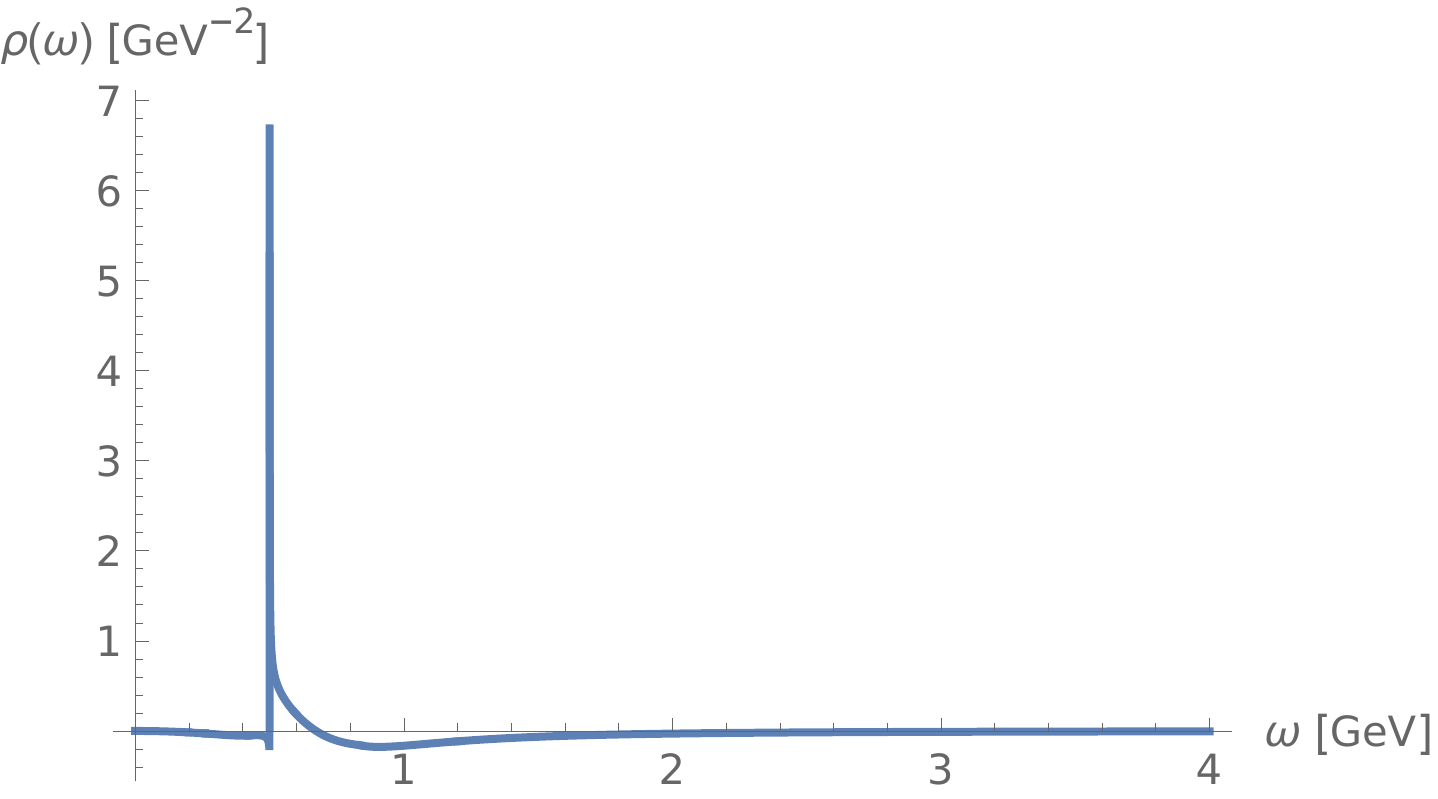}
  \end{center}
   \caption{An estimate of the spectral function if the almost real complex poles are artifacts of the approximation (\ref{eq:approx_oneloop}). This plots $\rho(\omega) = \frac{1}{\pi} \operatorname{Im} \mathscr{D}_T (k_0^2 = \omega^2 + i \epsilon')$ with $\epsilon'/M^2 =  10^{-3}$, which is larger than the imaginary part of the almost real pole. This shows that the spectral function has a long-lived quasiparticle peak, if the complex pole is an error.} 
    \label{fig:025spec_finitee.pdf}
\end{figure}

If the almost real pole is an artifact, an estimate of the spectral function will be given by $\rho(\omega) = \frac{1}{\pi} \operatorname{Im} \mathscr{D}_T (k_0^2 = \omega^2 + i \epsilon')$, where $\epsilon'$ is small but larger than the imaginary part of the almost real pole $\omega_I$. This estimate is displayed in Fig.~\ref{fig:025spec_finitee.pdf} at (\ref{eq:parameter_specific}) and $N_F = 2$. We take $\epsilon'/M^2 =  10^{-3}$ because the complex poles are at $k_0^2/M^2 \approx 1.4 \pm (1.8 \times 10^{-4}) i$. This plot implies that the new ``complex pole'' may correspond to actually a long-lived quasiparticle.
Such an appearance of a quasiparticle could be possibly related to the nuclear superfluidity.

Finally, we note that our results suggest that the $N_P = 4$ region is located in the region less than $\mu_q \approx 0.33$ GeV, which is approximately the matter threshold. This observation would imply that the new complex pole or the possible quasiparticle pole would be in the confined dynamics.
Moreover, it is curious that the right side of the boundary between the $N_P = 2$ and $N_P = 4$ regions locates near the liquid-gas threshold $\mu_q \approx 0.33$ GeV for all $m_q \lesssim 0.33$ GeV (and for $3 \lesssim g \lesssim 8$, see Appendix C).

In summary, we again emphasize the following points,
\begin{itemize}
    \item The chemical potential influences the gluon propagator significantly around $k_0 \approx 2 \mu_q$. This can be explained by the facts, (1) it is the least energy for the quark pair production without momentum transfer $\vec{k} = 0$ and (2) the quark loop is important in the energy scale less than the effective gluon mass in this model.
    \item If the new pair of complex conjugate poles indeed emerges as $\mu_q$ increases, there may be a transition on the boundary between the $N_P = 2$ and $N_P = 4$ regions.
    \item On the other hand, the almost real pole may be an artifact of the approximation (\ref{eq:approx_oneloop}). Then, the gluon propagator would have a quasiparticle spectral peak instead of the complex poles which correspond to confined states.
\end{itemize}

\subsection{RG improvement of computing $N_W(C)$}
As discussed in Sec.~V~A, the appearance of the new complex poles might be an artifact of the approximation.
Since the new complex poles have very small imaginary parts, it seems that the $N_P = 4$ region would be highly affected by a choice of approximation.
Therefore, it is essential to examine the robustness of the $N_P = 4$ region.
Here, as an attempt, we compute $N_W(C)$ from one-loop RG improved data for real frequencies $\{ \mathscr{D}_T (z^2 = x_n + i \epsilon) \}$.

Since we have adopted the renormalization scheme (\ref{eq:TWrenomalization}) described by the vacuum part, we can implement the RG improvement using the vacuum results \cite{PTW14}.
The infrared safety of the scheme (\ref{eq:TWrenomalization}) enables us to implement the one-loop RG improvement avoiding the Landau pole.
Note that the chemical potential $\mu_q$ does not run in this scheme.

The RG equation for the gluon propagator is given by
 \begin{align}
\mathscr{D}_{\mu \nu} &(k_E,\alpha(\mu^2),\mu_q;\mu^2) \notag \\
&= Z_A^{-1} (\mu^2,\mu_0^2) \mathscr{D}_{\mu \nu} (k_E,\alpha(\mu_0^2),\mu_q;\mu_0^2),
\end{align}
 where $\alpha$ denotes the set of gauge coupling and gluon and quark masses $\alpha = (g,M,m_q)$ and $Z_A (\mu^2,\mu_0^2)$ is the renormalization factor of the gluon field. Using this equation, we here approximate the gluon propagator for real frequencies as
  \begin{align}
\mathscr{D}_T &(k_0^2,\alpha(\mu_0^2),\mu_q;\mu_0^2) \notag \\
&\approx  Z_A (|k_0^2|,\mu_0^2) \mathscr{D}_T^{1-loop} (k_0^2,\alpha(|k_0^2|),\mu_q;|k_0^2|).
\end{align}

 \begin{figure}[tb]
  \begin{center}
   \includegraphics[width=0.9\linewidth]{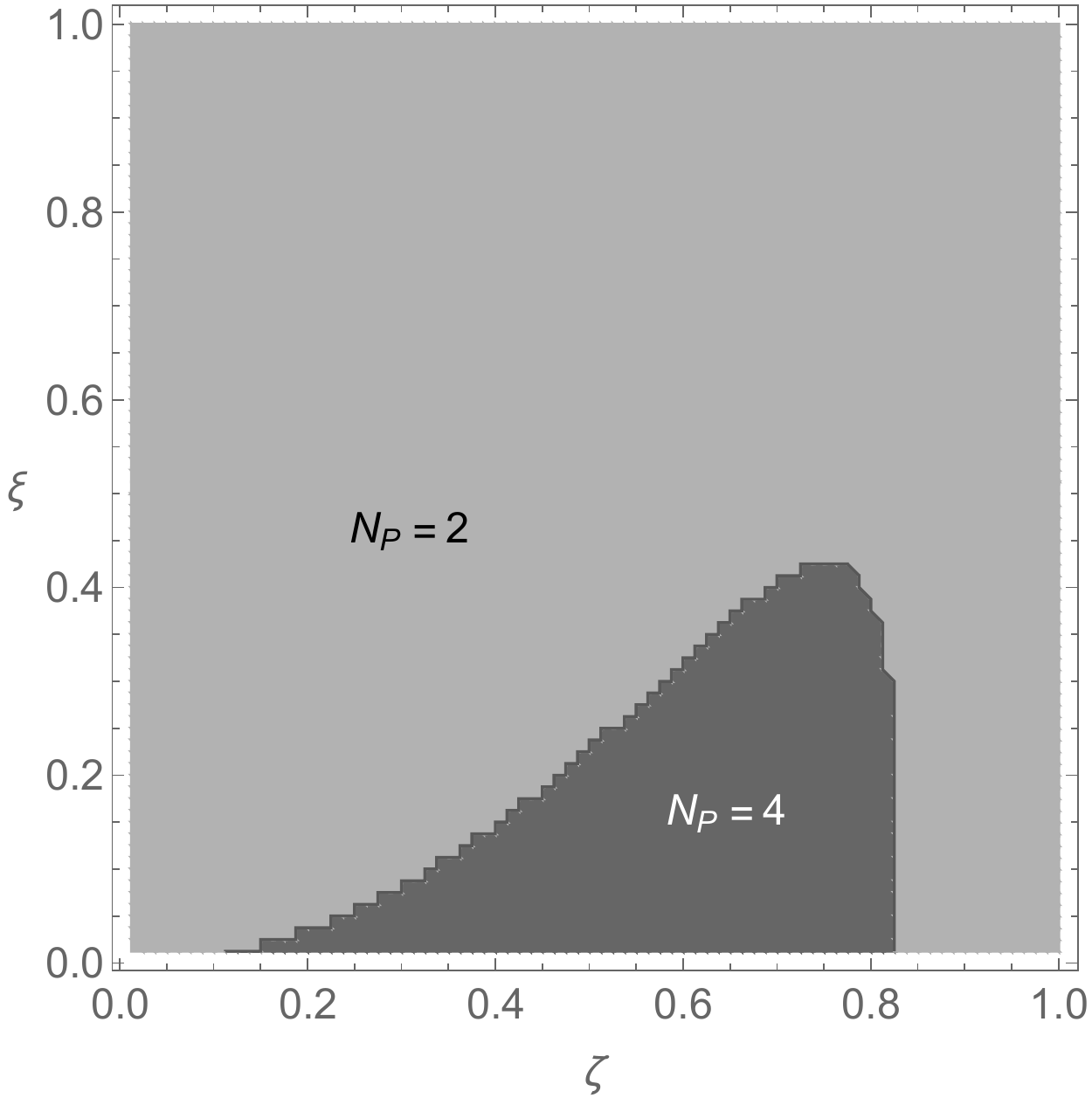}
  \end{center}
   \caption{Contour plot of $N_W(C)$ for the RG-improved gluon propagator on the $(\zeta = \frac{\mu_q^2}{M^2(\mu_0)},\xi= \frac{m_q^2(\mu_0)}{M^2(\mu_0)})$ plane at the set of parameters (\ref{eq:PTW_parameter}), which gives the number of complex poles through the relation $N_P = - N_W(C)$. In the $N_P = 2,4$ regions, the gluon propagator has one pair and two pairs of complex conjugate poles, respectively. We used $N+1 = 4 \times 10^5$, $x_0 = 4 \times 10^{-5}  ~ \mathrm{GeV^2} ,~~x_n - x_0 = n \times 10^{-5} ~ \mathrm{GeV^2} ~ (n = 1, \cdots, N+1)$ for the discretization (\ref{eq:winding-timelike}) and $\epsilon = 10^{-9} M^2$ for the infinitesimal imaginary part.
   }
    \label{fig:RG-winding}
\end{figure}

Here, we suppose $N_Z = 0$, since $ \mathscr{D}_T^{1-loop} (k_0^2)$ has no complex zeros. Then, we can compute the number of complex poles through the relation $N_P = - N_W(C)$.

A contour plot of $N_W(C)$ on the plane $(\mu_q^2, m_q^2 (\mu_0))$ at the set of parameters (\ref{eq:PTW_parameter}) computed by this approximation scheme is displayed in Fig.~\ref{fig:RG-winding}.
Note that the $N_P = 4$ region in Fig.~\ref{fig:RG-winding}
is qualitatively similar to that of Fig.~\ref{fig:typical_windingnumber}.
For example, at the set of parameters (\ref{eq:PTW_parameter}) and (\ref{eq:PTW_parameter_mq}), the $N_P = 4$ region lies in $0.24 ~\mathrm{GeV} \lesssim \mu_q \lesssim 0.38 ~\mathrm{GeV}$.
This qualitative similarity between Fig.~\ref{fig:typical_windingnumber} and Fig.~\ref{fig:RG-winding} suggests the robustness of the $N_P = 4$ region and that the gluon propagator could have additional complex poles rather than a quasiparticle pole in this region.

We ought to note that the one-loop RG improvement would still be not enough to capture important effects of quarks as pointed out in \cite{PTW14}.
Although another way of this examination is to include the two-loop quantum corrections, the two-loop calculation is beyond the scope of the present study.

\subsection{Slightly larger $\mu_q$}
To obtain a fair agreement with lattice results in two-color QCD, the effective gluon mass parameter $M$ is chosen of order $\mu_q$ for $\mu_q \sim 0.6$ -- $1$ GeV \cite{Suenaga-Kojo19}.
As an attempt to obtain an estimate of the analytic structure of the gluon propagator for the slightly large $\mu_q$, we investigate the strict one-loop gluon propagator at $\mu_q = M$.

 \begin{figure}[tb]
  \begin{center}
   \includegraphics[width=0.9\linewidth]{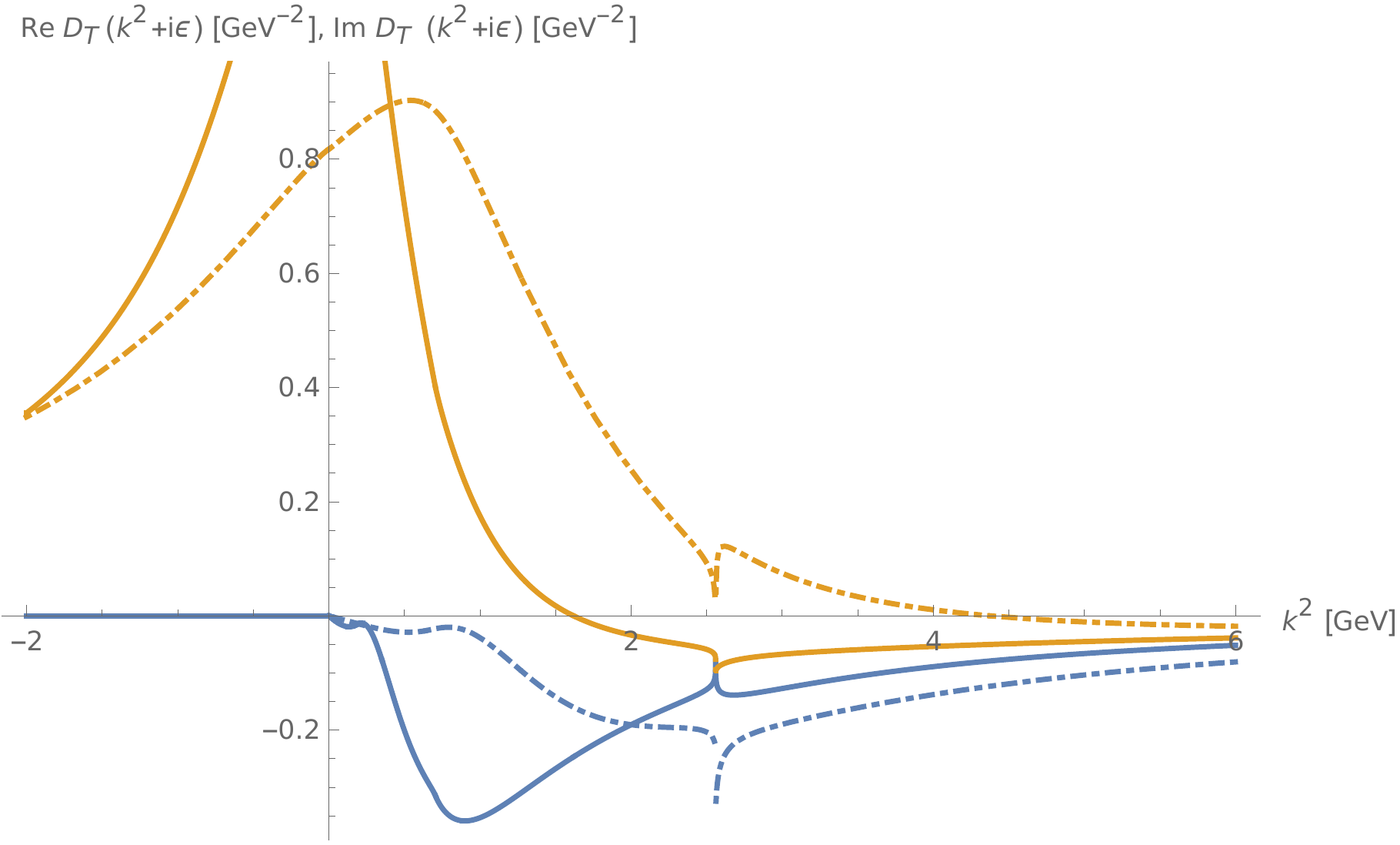}
  \end{center}
   \caption{The real and imaginary parts of the gluon propagator at $\mu_q = 0.8$ GeV for $M = 0.42$ GeV and $M = 0.8$ GeV. The solid curves are those of $M = 0.42$ GeV, which was regarded as the effective gluon mass in the vacuum. Those of $M = 0.8$ GeV are represented by dashed-dotted ones. The other parameters are $(g,m_q) = (4.5,0.13 \textrm{ GeV})$ and $N_F = 2$.} 
    \label{fig:inmedium_08gev.pdf}
\end{figure}

 \begin{figure}[tb]
  \begin{center}
   \includegraphics[width=0.9\linewidth]{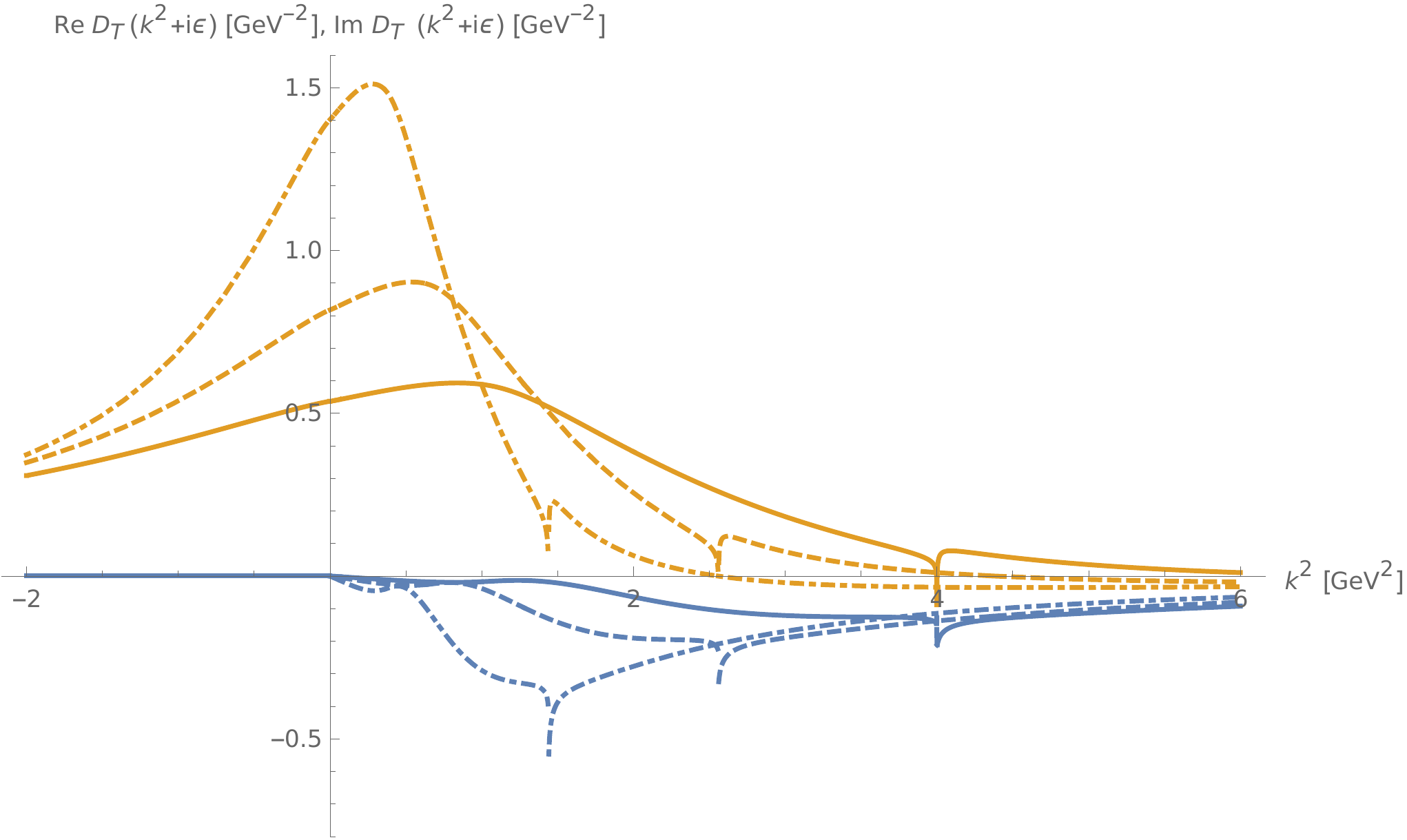}
  \end{center}
   \caption{The real and imaginary parts of the gluon propagator at $M=\mu_q = 0.6,0.8,1.0$ GeV are plotted as dashed-dotted, dashed, and solid curves, respectively. The other parameters are $(g,m_q) = (4.5,0.13 \textrm{ GeV})$ and $N_F = 2$.} 
    \label{fig:largemu1.pdf}
\end{figure}

Beforehand, let us see how the in-medium modification of the effective gluon mass affects the analytic structure.
The real and imaginary parts of the gluon propagator at $\mu_q = 0.8$ GeV for $M = 0.42 \mathrm{~GeV}$ and $M = 0.8 \mathrm{~GeV}$ are plotted in Fig.~\ref{fig:inmedium_08gev.pdf}.
The change of $M$ does not largely modify the location of the spectral peak, $k_0^2 \approx (2 \mu_q)^2$, while the direction of the peak is inverted.

The real and imaginary parts of the gluon propagator at $M=\mu_q = 0.6,0.8,1.0$ GeV are plotted in Fig.~\ref{fig:largemu1.pdf}. The spectral function has a negative peak at $k_0^2 \approx (2 \mu_q)^2$. The magnitude of this peak decreases as $\mu_q$ increases. The gluon propagator has one pair of complex conjugate poles as the vacuum one. The effect of the chemical potential around $k_0 \approx 2 \mu_q$ is less significant for large $\mu_q$ in this model.

For complex poles, we have numerically confirmed $N_P = 2$ in this set up for $\mu_q \sim 0.6$ -- $1$ GeV as inferred from Fig.~\ref{fig:typical_windingnumber}. Their positions $k_0 = \omega_R + i \omega_I,~(\omega_R > 0 , ~ \omega_I > 0)$ are plotted in Fig.~\ref{fig:largemu_estimate_pole_traj}. The apparent linearity of $\omega_R$ and $\omega_I$ with respect to $\mu_q (= M)$ suggests that $M$ and $\mu_q$ are the dominating scales in the propagator. A comparison with Fig.~\ref{fig:pole_traj} indicates that the in-medium modification of the gluon mass makes $\omega_R$ and $\omega_I$ larger. 

 \begin{figure}[tb]

 \begin{minipage}{\hsize}
  \begin{center}
   \includegraphics[width=0.9\linewidth]{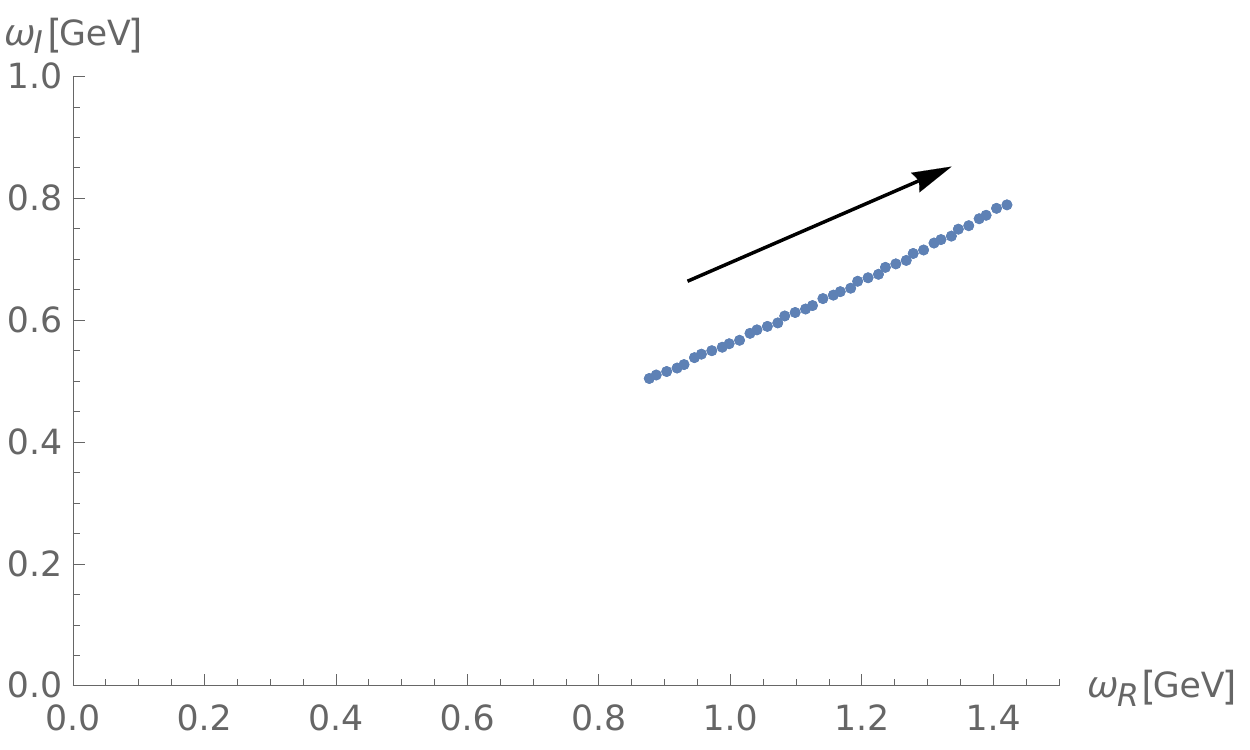}
  \end{center}
  \end{minipage}

 \begin{minipage}{\hsize}
  \begin{center}
   \includegraphics[width=0.9\linewidth]{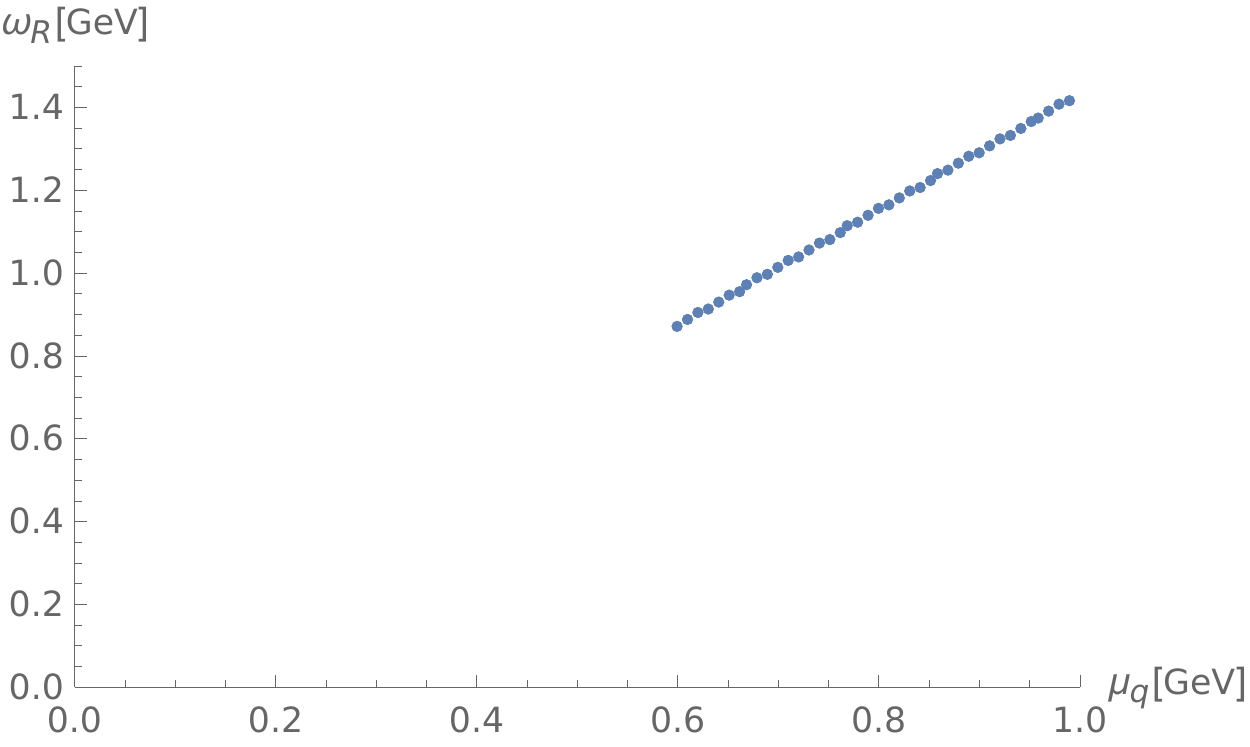}
  \end{center}
\end{minipage}

 \caption{Top panel: trajectory of a complex pole $k_0 = \omega_R + i \omega_I, ~ (\omega_R >0,~\omega_I > 0)$ of the gluon propagator in the plane $(\omega_R,\omega_I)$ at $(g,m_q) = (4.5, 0.13 \textrm{ GeV})$ with varying $\mu_q = M$ from 0.6 to 1 GeV. As $\mu_q$ increases, the pole moves along the arrow.
 Bottom panel: $\mu_q$ dependence of the real part $\omega_R$ of location of a complex pole.}
    \label{fig:largemu_estimate_pole_traj}
\end{figure}


\subsection{Comments on the infrared problems at finite temperature}
Finally, let us add remarks on the use of the naive perturbation theory in the massive Yang-Mills model in the thermal context.

We have investigated the gluon propagator by using the perturbation theory of the massive Yang-Mills theory.
However, it is well known that the naive perturbation theory of the Yang-Mills theory breaks down at finite-temperature. Therefore, some ressumation procedure is usually required to reach the infrared region in the thermal context, including cases at the finite chemical potential.
 Such a breakdown of the perturbation theory stems from the masslessness of the gluon. Indeed, due to the infrared singularities in the usual case, the generated magnetic gluon mass proportional to $g^2$ appears in denominators of terms of the perturbative series, which ruins the expansion in $g^2$. This breakdown is well known as the Linde problem \cite{Linde80}.

On the other hand, the massive Yang-Mills model does not present such a manifest breakdown of perturbation theory since the gluon mass regulates these infrared problems. Therefore, although resummation procedures may improve the results, the naive one-loop propagator of this model could already capture essential aspects of the Landau-gauge gluon propagator in QCD.

\section{Summary and future prospects}

Let us summarize our findings.
We have performed complex analyses of the gluon propagator at nonzero quark chemical potential $\mu_q$ in the long-wavelength limit $\vec{k} \rightarrow 0$, by using the massive Yang-Mills model.
We have verified that the two conditions, (i) $D(z) \rightarrow 0$ as $|z| \rightarrow \infty$ and (ii) $D(z)$ is holomorphic except for the real axis and a finite number of complex poles, are sufficient to single out the correct analytic continuation of a Matsubara propagator. Therefore, the uniqueness of the analytic continuation guaranteed in a similar sense as \cite{Baym-Mermin} even if we allow the existence of complex poles. For the proof, see Appendix A.

We have found that there is a $N_P = 4$ region, where the gluon propagator has two pairs of complex conjugate poles with respect to the complex variable $z^2 = k_0^2$. In this region, a new pair appears near the real axis in addition to the other pair similar to that in the vacuum case. At the typical parameters (Fig.~\ref{fig:typical_windingnumber}), the $N_P = 4$ region appears for light quarks ($m_q \lesssim 0.30$ GeV). As the quark chemical potential $\mu_q$ increases, the number of complex poles becomes four ($N_P = 4$) at slightly above the quark mass $m_q$ and backs to two ($N_P = 2$) at $\mu_q \approx 0.8 M \approx  0.33$ GeV. 
This structure is not sensitive to details of choice of the parameters $(g,M)$ as shown in Appendix C.
Moreover, in this $N_P = 4$ region, the new pair of complex conjugate poles has quite small imaginary part, and its location is approximately $\operatorname{Re} k_0^2 \approx (2 \mu_q)^2$.
On the other hand, in the $N_P = 2$ region, the gluon propagator behaves less ``particlelike'' with larger ratio $\omega_I/\omega_R$ of the complex pole at $k_0 = \omega_R + i \omega_I$, as $\mu_q$ increases.

The chemical potential influences the gluon propagator significantly around $k_0 \approx 2 \mu_q$, where the new poles appear and the spectral peak is observed. We can attribute this to the facts (i) it is the least energy for the quark pair production to occur at $\vec{k} = 0$ and (ii) the quark loop dominates in the energy scale less than the gluon mass $M$.

Finally, we can interpret the new almost real poles in two ways.
First, the results may imply that the gluon propagator indeed has a new pair of complex poles. This suggests a transition in confined degrees of freedom involving the gluon.
Second, the almost real pole may be an artifact of ``the one-loop approximation'' (\ref{eq:approx_oneloop}). Then, the gluon propagator would have a  long-lived quasiparticle spectral peak instead of the confined complex pole, which suggests a quasiparticle picture of the in-medium gluon.
Note that, however, the $N_P = 4$ region still remains in the RG-improved results (Fig.~\ref{fig:RG-winding}).

To sum up, although the gluon propagator presents only mild changes on the Euclidean side \cite{BHMS19}, it might have a rich and interesting structure in the complex frequency plane.

As future prospects, there is plenty of room for improvement in the present work in many aspects.
First, this work does not take into account the quark condensation, which is expected to be essential in the highly dense quark matter. The effect on the analytic structure of the quark gap would be interesting.  Second, as remarked in the introduction, the one-loop level is not enough in the quark sector of the massive Yang-Mills model. A possible improvement is the double expansion that improves the quark mass function significantly \cite{PRSTW17}.
Third, while a fair agreement with lattice results can be obtained by making the gluon mass $M$ depend on $\mu_q$ \cite{Suenaga-Kojo19}, the medium modification of the effective gluon mass should be determined in a more systematic way.
Fourth, since the massive Yang-Mills model has the infrared safe renormalization scheme, it would be important to compare the RG improved Euclidean gluon propagator with the lattice one. This could improve the current unsatisfactory agreement.
Lastly, when using lattice results, we have to keep in mind that the lattice gluon propagator has non-negligible systematic errors, e.g., finite lattice-spacing effect at low momenta \cite{DOS16}, and how Gribov copies affect results because there is no reason of the coincidence between the minimal Landau gauge and the Euclidean version of Landau gauge of the well-known covariant operator formalism due to the Gribov ambiguity. 

For other directions, it would be interesting to introduce temperature and to consider the physical sector and its transport properties in the massive Yang-Mills model and compare them with other approaches, e.g., \cite{Solana:2018pbk}.
Although it is very difficult, it is important to discuss implications of complex poles in the physical sector. The corresponding state should be confined and not itself have any physical impact, but its composite state might have physical significance \cite{BDGHSVZ10}. Formal aspects of complex poles will be discussed in a future work.

\section*{Acknowledgements}
We are grateful to Etsuko Itou and Daiki Suenaga for valuable comments.
Y.~H. is supported by JSPS Research Fellowship for Young Scientists Grant No.~20J20215, and K.-I.~K. is supported by Grant-in-Aid for Scientific Research, JSPS KAKENHI Grant (C) No.~19K03840.

\appendix

\section{Uniqueness of analytic continuation of the Matsubara propagator with complex poles}

In the absence of complex singularities, a theorem for the uniqueness of analytic continuation of the Matsubara propagator is well-known and proved in \cite{Baym-Mermin}. In this appendix, we shall extend the theorem to propagators with complex poles. 

In practice, we have not faced the problem of the uniqueness of the analytic continuation as we have employed the gluon propagator at zero temperature $T=0$. However, the propagator at finite chemical potential is the low temperature limit $T \rightarrow 0$ of the Matsubara propagator at finite temperature; it is conceptually essential to establish the uniqueness of the analytic continuation of the Matsubara propagator.\\

\begin{theorem}
Let $D(z)$ be a complex function whose values at Matsubara frequencies $z = i \omega_n := i \frac{2 \pi n}{\beta}$ are given. Then, its analytic continuation $D(z)$ to the whole complex $z$ plane is unique provided that an analytic continuation satisfies the following conditions,
\begin{enumerate}
    \item $D(z) \rightarrow 0$ if $|z| \rightarrow \infty$,
    \item $D(z)$ is holomorphic except for the real axis and a finite number of complex poles. 
\end{enumerate}
\end{theorem}

\begin{proof}


Let $D_1(z)$ and $D_2 (z)$ be two analytic continuations satisfying the above two conditions that coincide at all the Matsubara frequencies: $D_1 (i \omega_n) = D_2 (i \omega_n)$. Then, $\varphi(z) := D_1(z) - D_2(z)$ satisfies 
\begin{itemize}
    \item $\varphi(i \omega_n) = 0$ for all Matsubara frequencies $\omega_n$,
    \item $\varphi(z)$ is holomorphic except for the real axis and a finite number of complex poles,
    \item $\varphi(z) \rightarrow 0$ as $|z| \rightarrow \infty$.
\end{itemize}
We shall show that $\varphi(z)$ is identically zero, i.e., an assumption that $\varphi(z)$ had only isolated zeros leads to a contradiction. The proof is a straightforward generalization of a proof of the Carleman theorem given in Titchmarsh's book \cite{Titchmarsh}.

Consider the integral
\begin{align}
    I(R) := \oint_{C'} \frac{dz}{2 \pi i} \left( \frac{1}{R^2} - \frac{1}{z^2} \right) \ln \varphi(z + i \epsilon),
\end{align}
where the contour $C' = (\rho,R) \cup C_R \cup (-R,-\rho) \cup C_\rho$ is depicted in Fig.~\ref{fig:appendix_complex.pdf} and $C_R = \{ z ;~ \operatorname{Im}z >0, |z| = R \}$ and $C_\rho= \{ z ;~ \operatorname{Im}z >0, |z| = \rho \}$ are the semicircles with counterclockwise and clockwise directions respectively.

 \begin{figure}[t]
  \begin{center}
   \includegraphics[width=\linewidth]{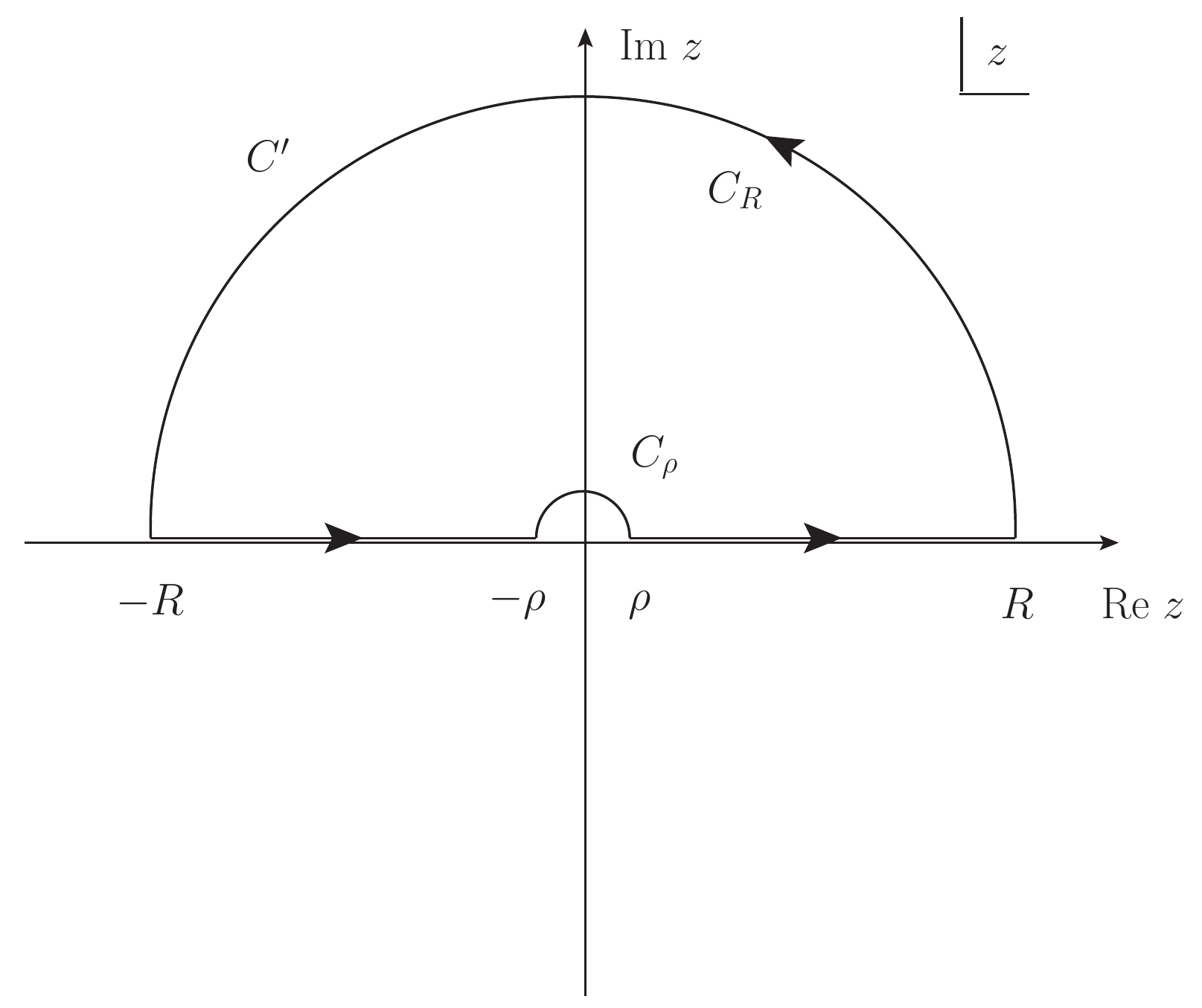}
  \end{center}
   \caption{The contour of the integral $I(R)$ consisting of lines $(-R,-\rho)$ and $(\rho,R)$ and semicircles $C_\rho$ and $C_R$: $C' = (\rho,R) \cup C_R \cup (-R,-\rho) \cup C_\rho$.} 
    \label{fig:appendix_complex.pdf}
\end{figure}

In this integral, we are going to keep $\rho$ finite and take a limit $R \rightarrow \infty$. From here on, we omit $+ i \epsilon$ for notational simplicity.

We take a sufficiently small $\rho$ (or appropriate choice of branch cuts of $\ln \varphi(z)$) so that $C_\rho$ does not intersect with any branch cut of the logarithm.

We evaluate this integral $\operatorname{Im} I(R)$ in two ways to obtain the contradiction.

First, we decompose the integral $I(R)$ into four pieces following $C' = (\rho,R) \cup C_R \cup (-R,-\rho) \cup C_\rho$,
\begin{align}
    I(R) = I_{\rho \rightarrow R} + I_{C_R} + I_{-R \rightarrow - \rho} + I_{C_\rho}.
\end{align}
Then, we have
\begin{align}
    I_{\rho \rightarrow R} + I_{-R \rightarrow - \rho} &= \int_\rho ^R \frac{dx}{2 \pi i} \left( \frac{1}{R^2} - \frac{1}{x^2} \right) \ln [ \varphi(x)\varphi(-x)],
\end{align}
and,
\begin{align}
    I_{C_R}  &= \int_{C_R} \frac{dz}{2 \pi i} \left( \frac{1}{R^2} - \frac{1}{z^2} \right) \ln \varphi(z) \notag \\
    &= \frac{i}{\pi R} \int_0 ^\pi d \theta ~ \sin \theta \ln \varphi(R e^{i\theta}).
\end{align}
Thus, we obtain
\begin{align}
    \operatorname{Im} I(R) &= \operatorname{Im} I_{C_\rho} \notag \\
    &+ \int_\rho ^R \frac{dx}{2 \pi } \left(  \frac{1}{x^2} - \frac{1}{R^2} \right) \ln | \varphi(x)\varphi(-x)| \notag \\
    &+ \frac{1}{\pi R} \int_0 ^\pi d \theta ~ \sin \theta \ln |\varphi(R e^{i\theta})|.
\end{align}
Note that $\operatorname{Im} I_{C_\rho}$ is $O(1)$ as $R \rightarrow \infty$. The other two integrals could diverge as $R \rightarrow \infty$; however, then, $\operatorname{Im} I(R)$ would be negative infinity, since $\varphi(z) \rightarrow 0$ as $|z| \rightarrow \infty$ and the other parts of the integrands are positive, $\left(  \frac{1}{x^2} - \frac{1}{R^2} \right)>0,~ \sin \theta > 0$. 
Therefore, $\operatorname{Im} I(R)$ is bounded from above: $\operatorname{Im} I(R) \leq M$ for some $M \in \mathbb{R}$.

On the other hand, the integral $I(R)$ is closely related to zeros and poles inside $C'$.
\begin{align}
    I(R) &= \oint_{C'} \frac{dz}{2 \pi i} \ln \varphi(z) \frac{d}{dz} \left( \frac{z}{R^2} + \frac{1}{z} \right) \notag \\
    &= \oint_{C'} \frac{dz}{2 \pi i} \frac{d}{dz} \left[ \ln \varphi(z)  \left( \frac{z}{R^2} + \frac{1}{z} \right) \right] \notag \\
    &- \oint_{C'} \frac{dz}{2 \pi i} \frac{\varphi'(z)}{\varphi(z)} \left( \frac{z}{R^2} + \frac{1}{z} \right). 
\end{align}

The first integral sums up ``discontinuities'' from the branch cuts of the logarithm.
Since we have assumed that the branch cuts of the logarithm do not intersect with $C_\rho$, the first term contributes only from $(\rho,R) \cup C_R \cup (-R,-\rho)$, on which $\left( \frac{z}{R^2} + \frac{1}{z} \right)$ is real.
Therefore,
\begin{align}
    \operatorname{Im} \oint_{C'} \frac{dz}{2 \pi i} \frac{d}{dz} \left[ \ln \varphi(z)  \left( \frac{z}{R^2} + \frac{1}{z} \right) \right] = 0. 
\end{align}

Finally, the second term can be evaluated as a weighted sum of zeros and poles. The generalized argument principle yields
\begin{align}
    - \oint_{C'} &\frac{dz}{2 \pi i} \frac{\varphi'(z)}{\varphi(z)} \left( \frac{z}{R^2} + \frac{1}{z} \right) \notag \\
    &= - \sum_{\substack{z_j:\mathrm{zeros} \\ z_j \in \mathcal{D}'} } \left( \frac{z_j}{R^2} + \frac{1}{z_j} \right) + \sum_{\substack{w_k:\mathrm{poles} \\ w_k \in \mathcal{D}'} } \left( \frac{w_k}{R^2} + \frac{1}{w_k} \right),
\end{align}
where $\mathcal{D}'$ is the region surrounded by $C'$.
To sum up,
\begin{align}
    \operatorname{Im} I(R) &= \sum_{\substack{z_j:\mathrm{zeros} \\ z_j \in \mathcal{D}'} } \left( \frac{1}{r_j} - \frac{r_j}{R^2} \right) \sin \theta_j + O(1),
\end{align}
where we have defined $r_j e^{i \theta_j} := z_j$, used the finiteness of the number of poles, and $O(1)$ stands for a finite term for all $R$. As $\varphi(i \omega_n) = 0$ for all Matsubara frequencies and $ \left( \frac{1}{r_j} - \frac{r_j}{R^2} \right) > 0$ for $z_j \in \mathcal{D}'$, 
\begin{align}
    \sum_{\substack{z_j:\mathrm{zeros} \\ z_j \in \mathcal{D}'} } \left( \frac{1}{r_j} - \frac{r_j}{R^2} \right) \sin \theta_j \geq \sum_{\substack{n \\ i \omega_n \in \mathcal{D}'} } \left( \frac{1}{\omega_n} - \frac{\omega_n}{R^2} \right)
\end{align}
Moreover, as $R \rightarrow \infty$,
\begin{align}
    \sum_{\substack{n \\ i \omega_n \in \mathcal{D}'} }  \frac{\omega_n}{R^2} = \sum_{0< \omega_n < R }  \frac{\omega_n}{R^2} = O(1).
\end{align}
These results indicate
\begin{align}
    \operatorname{Im} I(R) &> \sum_{0< \omega_n < R } \frac{1}{\omega_n}+ O(1) \rightarrow +\infty,
\end{align}
which contradicts the first evaluation: $\operatorname{Im} I(R)$ is bounded above. The assumption that $\varphi(z)$ had only isolated zeros is false. Therefore, $\varphi(z) = D_1(z) - D_2(z)$ is identically zero at least for the upper-half plane. In the same way, $\varphi = 0$ in the lowest-half plane follows by taking $-z$ as $z$. This completes the proof.
\end{proof}

Incidentally, let us comment on the possibility of branch cuts.
The uniqueness holds even if we allow the propagator to have a finite number of (nonclosed) branch cuts that have finite length and represent finite discontinuities of $D(z)$.
Then, $\varphi(z)= D_1(z) - D_2(z)$ could have branch cuts in addition to poles. We can still prove $\varphi = 0$ by (i) deforming $C'$ to avoid the branch cuts and (ii) taking the branch cuts of $\ln \varphi(z)$ so that they intersect with neither $C_\rho$ nor the path wrapping around the new branch cuts of $\varphi$.

Indeed, the first evaluation becomes
\begin{align}
    I(R) &= I_{\rho \rightarrow R} + I_{C_R} + I_{-R \rightarrow - \rho} + I_{C_\rho} + \sum_{\gamma: cuts} I_\gamma, \notag \\
    I_\gamma &= \int_{\gamma'} \frac{dz}{2 \pi i} \left( \frac{1}{R^2} - \frac{1}{z^2} \right) \ln \varphi(z),
\end{align}
where $\gamma'$ is a path that surrounds a cut $\gamma$ in $|z| < R$.
This new contribution is finite for any $R$ due to the finiteness of the branch cuts. 

On the other hand, the second evaluation by the partial integration is the same as before, which leads to a contradiction again. Therefore, the conclusion is not changed in the presence of discontinuities on curves of finite length.

\section{One-loop expressions for the vacuum part}
Here, we present the one-loop expression for $\Pi^{vac}_{\mu \nu} (k_E) = \Pi^{vac} (k_E^2) P_{\mu \nu} = M^2 \hat{\Pi}^{vac} (\frac{k_E^2}{M^2}) P_{\mu \nu}$.

Beforehand, we rewrite the two-point vertex functions $\Gamma_{{\mathscr A},vac}^{(2)}$ and $\Gamma_{gh,vac}^{(2)}$ by dimensionless gluon and ghost vacuum polarizations $\hat{\Pi}$ and $\hat{\Pi}_{gh}$ as 
\begin{align}
\Gamma_{{\mathscr A}, vac}^{(2)}(k_E^2) &= M^2 [s+1 + \hat{\Pi}(s) + s \delta_Z + \delta_{M^2}] \notag \\
&=: M^2 [s+1 + \hat{\Pi}^{vac}(s)], \label{eq:gluon_vertex_strict_one_loop}\\
\Gamma_{gh, vac}^{(2)}(k_E^2) &:= - [\Delta_{gh} (k_E^2)]^{-1}  \notag \\
&= M^2 [s + \hat{\Pi}_{gh}(s) + s \delta_C] \notag \\
&=: M^2 [s + \hat{\Pi}_{gh}^{ren}(s) ],
\end{align}
where $k_E$ is the Euclidean momentum, $s = \frac{k_E^2}{M^2}$, and $\delta_Z:= Z_A - 1$, $\delta_{M^2} := Z_A Z_{M^2} -1 $, and $\delta_C:= Z_C -1$ are the counterterms.

The bare vacuum polarizations computed by the dimensional regularization read \cite{TW11,PTW14}, for gluons,
\begin{align}
\hat{\Pi}(s) &=\hat{\Pi}_{YM}(s) + \hat{\Pi}_{q}(s), \\ 
\hat{\Pi}_{YM}(s) &= \frac{g^2 C_2 (G)}{192 \pi^2} s \Biggl\{ \left( \frac{9}{s} - 26 \right) \left[ \varepsilon^{-1} + \ln (\frac{4 \pi}{M^2 e^\gamma})  \right] \notag \\
&- \frac{121}{3} + \frac{63}{s} + h(s) \Biggr\}, \notag \\
\hat{\Pi}_{q}(s) &= - \frac{g^2 C(r)}{6 \pi^2} s \Biggl\{ - \frac{1}{2} \left[ \varepsilon^{-1} + \ln (\frac{4 \pi}{m_q^2 e^\gamma}) \right] \notag \\
&- \frac{5}{6} + h_q \left( \frac{\xi}{s} \right)  \Biggr\},
\end{align}
for ghosts,
\begin{align}
\hat{\Pi}_{gh}(s) &= \frac{g^2 C_2 (G)}{64 \pi^2} s \biggl[ -3 \left[ \varepsilon^{-1} + \ln (\frac{4 \pi}{M^2 e^\gamma}) \right] \notag \\
& -5 + f(s) \biggr], 
\end{align}
where $\varepsilon := 2 - D/2$, $\gamma$ is the Euler-Mascheroni constant, $C_2(G)$ and $C(r) = N_F/2$ are the Casimir invariants of the adjoint and fundamental (with multiplicity $N_F$) representations of the gauge group $G$, $\xi = \frac{m_q^2}{M^2}$ and,
\begin{align}
h(s) &:= - \frac{1}{s^2} + \left( 1- \frac{s^2}{2} \right) \ln s \notag \\
&+ \left( 1+ \frac{1}{s}\right)^3 (s^2 - 10s + 1) \ln (s+1) \notag \\
&+ \frac{1}{2} \left( 1+ \frac{4}{s} \right)^{3/2} (s^2 - 20 s + 12) \ln \left( \frac{\sqrt{4+s} - \sqrt{s}}{\sqrt{4+s} + \sqrt{s}} \right), \notag \\
h_q(\tilde{t}) &:= 2\tilde{t} + (1-2 \tilde{t}) \sqrt{4 \tilde{t} + 1} \coth^{-1} (\sqrt{4 \tilde{t} +1}), \notag \\
   f(s) &:= - \frac{1}{s} - s \ln s + \frac{(1+s)^3}{s^2} \ln (s+1),
\end{align}
with $\tilde{t} := \frac{\xi}{s} = \frac{m_q^2}{k_E^2}$.

The renormalization conditions (\ref{eq:TWrenomalization}) for the gluon and ghost sector can be cast into in the one-loop level,
\begin{align}
 \begin{cases}
 Z_A Z_C Z_{M^2} = 1 \\
 \Gamma_{{\mathscr A},vac}^{(2)} (k_E^2 = \mu^2) = \mu^2 + M^2\\
 \Gamma_{gh,vac}^{(2)}(k_E^2 = \mu^2) = \mu^2 
 \end{cases}
 \Leftrightarrow \ 
 \begin{cases}
 \delta_C + \delta_{M^2} = 0 \\
 \hat{\Pi}^{vac}(s = \nu) = 0 \\
 \hat{\Pi}_{gh}^{ren} (s = \nu) = 0,
 \end{cases}
\end{align}
with $\nu := \frac{\mu^2}{M^2}$.

By imposing this renormalization condition, we have the renormalized two-point vertex functions,
\begin{align}
\hat{\Pi}^{vac}(s) &=\hat{\Pi}_{YM}^{ren.}(s) + \hat{\Pi}_{q}^{ren.}(s), \label{eq:vacuum_pol_TW}\\
\hat{\Pi}^{ren.}_{YM}(s) &= \frac{g^2 C_2 (G)}{192 \pi^2} s \biggl[ \frac{48}{s} + h(s) + \frac{3 f(\nu)}{s} - (s \rightarrow \nu) \biggr], \label{eq:vacuum_pol_TW_YM} \\
\hat{\Pi}^{ren.}_{q}(s) &=  - \frac{g^2 C(r)}{6 \pi^2} s \left[ h_q \left( \frac{\xi}{s} \right) - h_q \left( \frac{\xi}{\nu} \right)  \right].
\end{align}
Note that the gluon propagator at $T = \mu = 0$ exhibits the decoupling feature and satisfies the condition (ii) of Sec.~II~B:
\begin{align}
\hat{\Pi}^{vac}(s = 0) &> 0, \label{eq:pos_IR_prop_cond}  \\
\Rightarrow \Gamma_{{\mathscr A},vac}^{(2)}& (k_E^2 = 0) = M^2 [1 + \hat{\Pi}^{vac} (0) ] > 0. 
\end{align}
Indeed, we have
\begin{align}
\hat{\Pi}^{ren.}_{YM}(s=0) &= \frac{g^2 C_2 (G)}{192 \pi^2} \Biggl[ 3 f(\nu) - \frac{15}{2} \Biggr]>0, \notag \\
\hat{\Pi}^{ren.}_{q}(s=0) &= 0, 
\end{align}
where we have used $h_q(\tilde{t} \rightarrow \infty) = O(1)$, $h(s) = - \frac{111}{2s} + O(\ln s)$, $f(0) = 5/2$, and the fact that $f(s)$ increases monotonically in $s$.

Note also that the strict one-loop expression has the following asymptotic form in the limit $|k^2| \rightarrow \infty$:
 \begin{align}
\Gamma_{{\mathscr A},vac}^{(2)} \simeq g^2 \gamma_0 (-k^2) \ln |k^2| + O(k^2), \label{eq:gluon_asymptotic_UV_one_loop}
\end{align}
while the asymptotic freedom and RG analysis yields
\begin{align}
\Gamma_{{\mathscr A},vac}^{(2)} \simeq Z_{UV}^{-1}(-k^2) (\ln |k^2|)^{\gamma_0/\beta_0},
\end{align}
where we have analytically continued the gluon propagator from the Euclidean momentum $k^2 = - k_E^2$ to complex $k^2$, $Z_{UV} > 0$ is a positive constant, and $\gamma_0$ and $\beta_0$ are respectively the first coefficients of the gluon anomalous dimension and the beta function:
\begin{align}
\gamma_0 &= - \frac{1}{16 \pi^2} \left( \frac{13}{6} C_2 (G) - \frac{4}{3} C(r) \right),\notag \\
\beta_0 &= - \frac{1}{16 \pi^2} \left( \frac{11}{3} C_2 (G) - \frac{4}{3} C(r) \right).
\end{align}
Both the strict one-loop gluon propagator and RG improved one satisfy the condition (i) of Sec.~II~B.
In spite of the wrong logarithmic exponent, the one-loop gluon propagator has qualitatively the same phase as the RG improved one (for $N_F < 10$).
Thus, the wrong logarithmic exponent will not change the value of $N_W(C) = N_Z - N_P$, and hence the strict one-loop expression may be enough for our purpose.

\section{Number of complex poles with various $(g,M)$}

In the main text, we have investigated the analytic structure of the gluon propagator with the fixed parameters $g = 4.5$ and $M = 0.42$ GeV, as they give best-fit parameters to the lattice results \cite{PTW14}. In this appendix, we check that the qualitative features of the analytic structure are not sensitive to the model parameters $(g,M)$.

We have confirmed that the contour plots of $N_P$ on the $(\zeta = \frac{\mu_q^2}{M^2}, \xi = \frac{m_q^2}{M^2})$ plane are qualitatively same. Indeed, Fig.~\ref{fig:different_g} gives contour plots of $N_P$ at $M = 0.42$ GeV and $g = 3$ (top) and $g = 8$ (bottom). Figure~\ref{fig:different_M} gives contour plots of $N_P$ at $g = 4.5$ and $M = 0.3$ GeV (top) and $M = 0.8$ GeV (bottom).
The setup of the numerical calculations is the same as Fig.~\ref{fig:typical_windingnumber}. 
Similar to Fig.~\ref{fig:typical_windingnumber}, the left boundary (small-$\zeta$ side of the boundary) is near $\mu_q \sim m_q$ and the right boundary (large-$\zeta$ side of the boundary) at $\mu_q^2 \approx 0.6 M^2$, at least within the parameter region $3 \lesssim g \lesssim 8$ and $ 0.3 \mathrm{~GeV} \lesssim M \lesssim 0.8$ GeV.

 \begin{figure}[t]
 
  \begin{minipage}{\hsize}
  \begin{center}
   \includegraphics[width=0.8\linewidth]{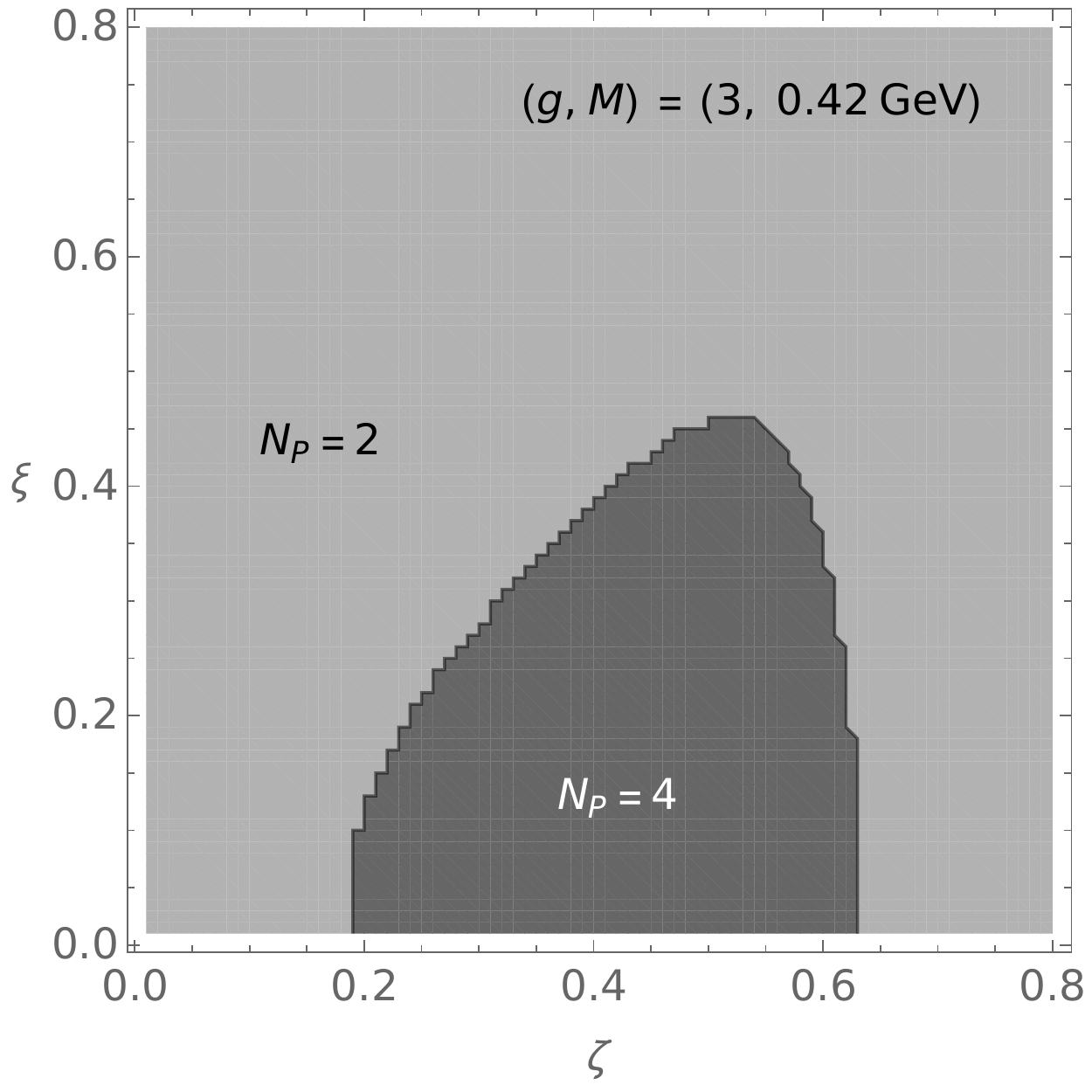}
  \end{center}
 \end{minipage}
 
  \begin{minipage}{\hsize}
  \begin{center}
   \includegraphics[width=0.8\linewidth]{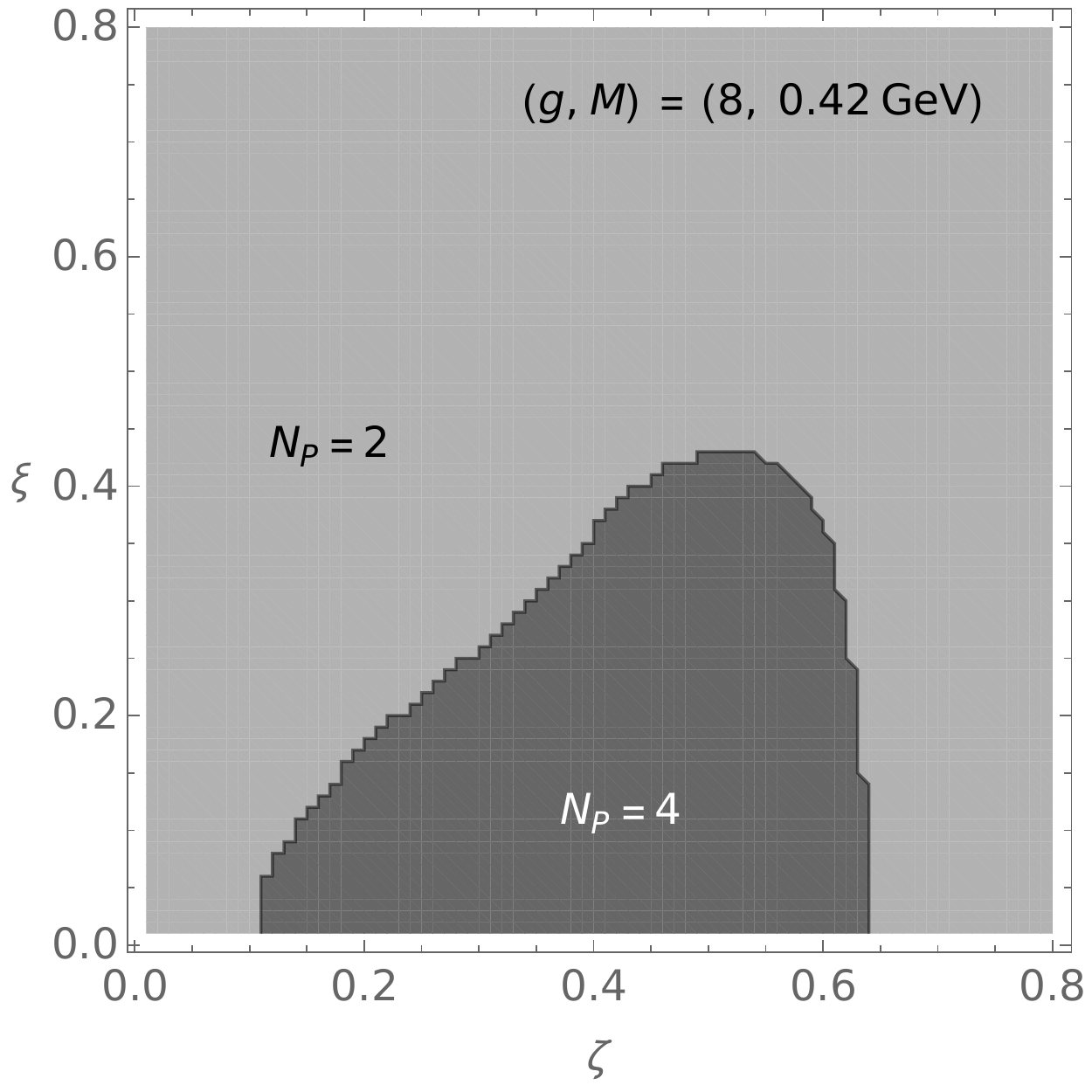}
  \end{center}
 \end{minipage}

 \caption{Contour plots of $N_W(C) = -N_P$ of the gluon propagator on $(\zeta = \frac{\mu_q^2}{M^2}, \xi = \frac{m_q^2}{M^2})$ plane at $g = 3$ (top) and $g = 8$ (bottom) for $M = 0.42$ GeV.
 }
    \label{fig:different_g}
\end{figure}

 \begin{figure}[t]
 
  \begin{minipage}{\hsize}
  \begin{center}
   \includegraphics[width=0.8\linewidth]{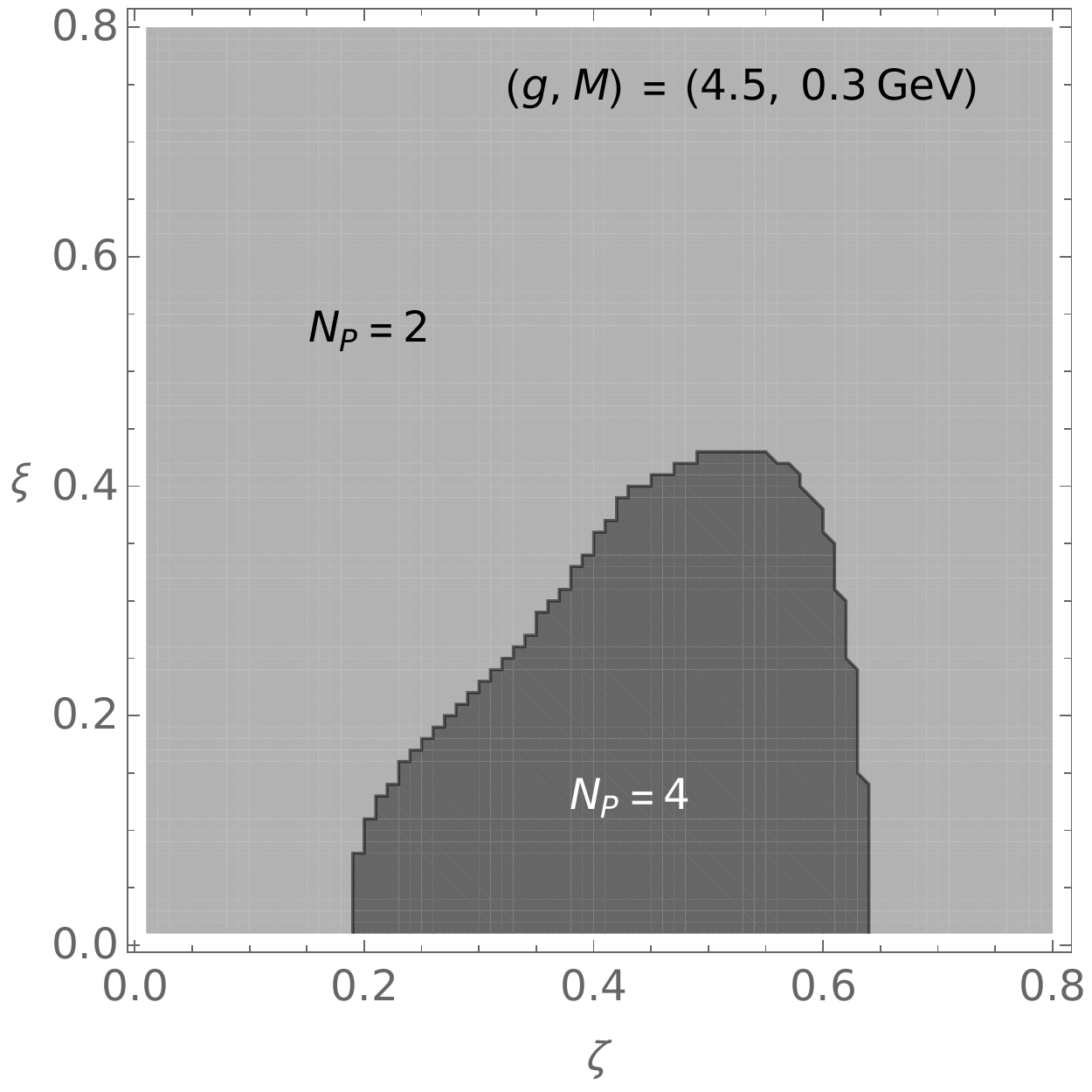}
  \end{center}
 \end{minipage}
 
  \begin{minipage}{\hsize}
  \begin{center}
   \includegraphics[width=0.8\linewidth]{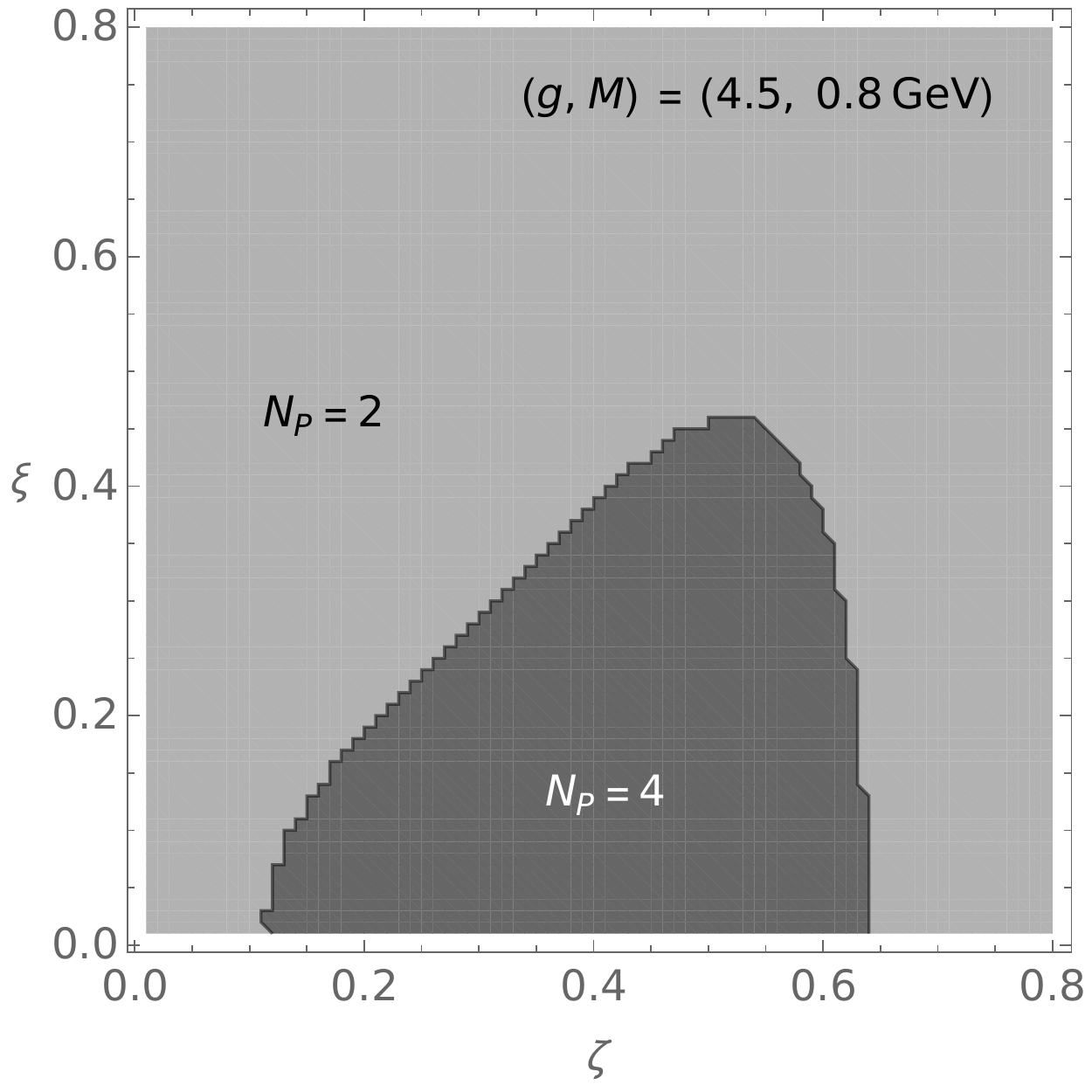}
  \end{center}
 \end{minipage}

 \caption{Contour plots of $N_W(C) = -N_P$ of the gluon propagator on $(\zeta = \frac{\mu_q^2}{M^2}, \xi = \frac{m_q^2}{M^2})$ plane at $M = 0.3$ GeV (top) and $M = 0.8$ GeV (bottom) for $g = 4.5$.}
    \label{fig:different_M}
\end{figure}

\end{document}